\documentclass[runningheads]{llncs}
\usepackage[T1]{fontenc}
\usepackage{orcidlink}
\usepackage{graphicx}
\usepackage{mathtools}
\usepackage{xcolor}
\usepackage{wrapfig}
\usepackage{cancel}
\usepackage{listings}
\usepackage{amsmath}
\usepackage{amssymb}
\usepackage[capitalise]{cleveref}
\usepackage[normalem]{ulem}
\usepackage{etoolbox}
\usepackage[vskip=3.5pt]{quoting}
\usepackage{subcaption}

\usepackage{xspace}

	\newif\ifsubmit
	\submittrue
	\ifsubmit 
	  \newcommand{\FDComm}[1]{\relax} 
	  \newcommand{\VPComm}[1]{\relax} 
	  \newcommand{\UTComm}[1]{\relax} 
	   \newcommand{\PAComm}[1]{\relax} 
           
           \def\bfd{}
	  \def\efd{} 
	  \def\bpa{}
	  \def\epa{} 
	  \def\but{} 
	  \def\eut{} 
	\else
	 \newcommand{\FDComm}[1]{\textcolor{red}{\scriptsize [FD: #1]}}
	
	\newcommand{\PAComm}[1]{{\scriptsize \textcolor{magenta}{[Paola{:} #1]}}}
        \newcommand{\VPComm}[1]{{\scriptsize \textcolor{violet}{[VP{:} #1]}}}
        \newcommand{\UTComm}[1]{{\scriptsize \textcolor{blue}{[UT{:} #1]}}}
        
        \def\bfd{\begin{color}{red}}
	  \def\efd{\end{color}} 
	  \def\bpa{\begin{color}{magenta}} 
	  \def\epa{\end{color}} 
	  \def\bez{\begin{color}{blue}} 
	  \def\eez{\end{color}} 
	  \def\but{\begin{color}{blue}} 
	  \def\eut{\end{color}} 
		
	\fi

	\definecolor{lightgray}{gray}{0.9}
	\definecolor{airforceblue}{rgb}{0.36, 0.54, 0.66}
	\definecolor{bittersweet}{rgb}{1.0, 0.44, 0.37}
	\definecolor{brilliantlavender}{rgb}{0.96, 0.73, 1.0}

\newcommand{\refItem}[2]{\cref{#1}(\ref{#1:#2})}

\newcommand{\ple}[1]{\ensuremath{\langle #1 \rangle}\xspace} 
\newcommand{\fun}[3]{#1 : #2\to#3}
\newcommand{\pfun}[3]{#1 : #2 \rightharpoonup #3}

	\newcommand{\resid}{\mathit{r}}
	\newcommand{\actorid}{\mathit{a}}
	\newcommand{\aactorid}{\mathit{b}}
	\newcommand{\methid}{\mathit{m}}
	\newcommand{\futid}{\mathit{f}}
	\newcommand{\gradeid}{\mathit{g}}
\newcommand{\gr}{\gradeid}
	\newcommand{\ggradeid}{\mathit{h}}
\newcommand{\ggr}{\ggradeid}

	\newcommand{\var}{\mathit{x}}
	\newcommand{\vvar}{\mathit{y}}

	\newcommand{\val}{\mathit{v}}

	\newcommand{\threadid}{\mathcal{T}}

	\newcommand{\GMon}{{\mathit{G}}} 
	\newcommand{\GSet}{|\GMon|}  
	\newcommand{\gsum}{+}
	
	\newcommand{\gord}{\leq}
	\newcommand{\gsub}{-}
	\newcommand{\gzero}{\mathbf{0}}

    	\newcommand{\key}[1]{\mbox{\fontsize{8}{8}\selectfont\sffamily\bfseries #1}}
	
	\newcommand{\define}{::=}
	\newcommand{\many}[1]{\overline{#1}}

	\newcommand{\T}{\ensuremath \mathit{T}} 
	\newcommand{\FutTy}[2]{\ensuremath \key{Fut}\langle#1,#2\rangle} 

	\newcommand{\valueexpression}{\mathit{ve}}
	\newcommand{\expression}{\mathit{e}}
         
	\newcommand{\res}{\gradedr{\resid}{\gradeid}}
	\newcommand{\rres}{\gradedr{\resid}{\ggradeid}}
	
	\newcommand{\letin}[3]{\key{let} \; #1 \; \key{=} \;#2 \; \key{in} \; #3}

	\newcommand{\return}[1]{\key{return} \; #1}
	\renewcommand{\return}[1]{\key{ret} \; #1}
	\newcommand{\magiccall}[3]{#1!#2(#3)}

	\newcommand{\holdexp}[2]{\key{hold} \; #1\; #2}
	\newcommand{\releaseexp}[2]{\key{release} \; #1\; #2}
	\newcommand{\awaitexp}[1]{#1?}
	\newcommand{\callexp}[3]{#1!#2(#3)}
	\newcommand{\Sigop}{\mathcal{S}_{\key{op}}} 
	\newcommand{\opexp}[1]{\key{op}(#1)}
	\newcommand{\letexp}[3]{\letin{#1}{#2}{#3}}

	\newcommand{\unit}{\key{unit}}
	\newcommand{\Unit}{\key{Unit}}
	\newcommand{\actorenvironment}{\actorcontext} 
	\newcommand{\aactorenvironment}{\aactorcontext} 
	\newcommand{\processid}{P}
	\newcommand{\pprocessid}{Q}
	\newcommand{\ppprocessid}{U}
	\newcommand{\pdone}{\aux{done}} 
	\newcommand{\parP}{\mathrel{\|}}
	\newcommand{\parallelcomp}[2]{#1 \parP#2}

	\newcommand{\vestep}{\Rightarrow}
	\newcommand{\expstep}[2][]{\xrightarrow[#1]{#2}}

	\newcommand{\veeval}[2]{#1\; | \; #2}
	\newcommand{\expeval}[2]{#1\; | \; #2}

	\newcommand{\fp}[1]{\mathit{fp(#1)}}
	\newcommand{\fr}[1]{\mathit{fr(#1)}}

	\newcommand{\labelid}{l}
	\newcommand{\emptylabel}{\tau}

        \newcommand{\ch}{\oplus}
	\newcommand{\nchoicesym}{\mathrel{\ch}}
	\newcommand{\nondetchoic}[2]{#1\nchoicesym#2}

	\newcommand{\gradedvar}[2]{\gradedr{\mathit{#1}}{#2}}

	\newcommand{\dis}[1]{dis(#1)}
	\newcommand{\Subst}[3]{#1[#3/#2]} 

	\newcommand{\config}{\sigma}

	\newcommand{\resenv}{\rho}
	
	\newcommand{\activationrecord}{\lambda}

	\newcommand{\varbind}[3]{#1[ #3/ #2]}
	
	\newcommand{\assocval}[2]{#1 \mapsto #2}

	\newcommand{\dom}{\aux{dom}}

	\newcommand{\markcontext}{\aux{mark}}
	\newcommand{\fresh}{\aux{fresh}}

	\newcommand{\gradedr}[2]{#1^{#2}}
	\newcommand{\idle}{\key{idle}}
	\newcommand{\anythread}[5][]{#5\ifblank{#1}{}{^{#1}}[#2 \mid #3]^{#4}}
	\newcommand{\ath}{\bullet}
	\newcommand{\qth}{\circ} 
	\newcommand{\queuedthread}[4]{\anythread[\qth]{#1}{#2}{#3}{#4}} 
        \newcommand{\activethread}[4]{\anythread[\ath]{#1}{#2}{#3}{#4}} 
	\newcommand{\idleactor}[1]{\idle^{\mathit{#1}}}
      
	\newcommand{\placefut}[2]{#1 \leftarrow #2}

	\newcommand{\callmsg}[4]{#1 \leftarrow \magiccall{#2}{#3}{#4}}
	
	\newcommand{\holdmsg}[1]{\key{hold} \; #1}
	\newcommand{\rlsmsg}[1]{\key{rls} \; #1}

	\newcommand{\semanticstep}{\longrightarrow}
        \newcommand{\red}{\semanticstep}
	\newcommand{\manysemanticsteps}{\longrightarrow^*}

	\newcommand{\emptycontext}{\emptyset}

	\newcommand{\actorcontext}{\Phi}      
	\newcommand{\aactorcontext}{\Psi}    
	\newcommand{\aaactorcontext}{\Theta}
	\newcommand{\typingcontext}{\Gamma} 
	\newcommand{\ttypingcontext}{\Delta}    
	\newcommand{\futcontext}{\Sigma}
	\newcommand{\ffutcontext}{\Omega} 

	\newcommand{\rescontext}{\Theta}

	\newcommand{\joinctx}{; \;}

	\newcommand{\aux}[1]{\mathsf{#1}}       

	\newcommand{\linearsum}{,}
       \newcommand{\tysum}{\gsum} 
        \newcommand{\actsum}{\gsum}

           \newcommand{\actord}{\gord}

	\newcommand{\assoctype}[2]{#1 : #2}
	\newcommand{\assoctypeS}[2]{#1{:} #2}

	\newcommand{\assocctx}[2]{#1 :: #2}
        \newcommand{\assocctxS}[2]{#1{::}#2}

	\newcommand{\assocresenv}[2]{#1:#2}

\renewcommand{\key}[1]{\texttt{#1}}   

\newcommand{\ve}{\valueexpression}   
\newcommand{\e}{\expression}                
\newcommand{\PP}{\processid}                
\newcommand{\PQ}{\pprocessid}                
\newcommand{\Conf}[2]{#1\; ||\; #2}       
\newcommand{\defin}[1]{\widetilde{#1}}      

\newcommand{\Pcong}{\triangleright}
\newcommand{\Cong}{\bowtie}

\newcommand{\mbody}[2]{\aux{mbody}(#1,#2)}
\newcommand{\mtype}[2]{\aux{mtype}(#1,#2)}
\newcommand{\MBody}[2]{(#1,#2)}
\newcommand{\MType}[5]{(#1,#2\to#3,#4,#5)}
\newcommand{\OpTy}[3]{#1:#2\to#3}

\newcommand{\WTexpr}[8]{#1 \joinctx #2 \joinctx #3 \vdash^{#4}_{#5} \assoctype{#6}{#7\joinctx #8}}
\newcommand{\WTexprS}[8]{#1;#2 ;#3 \vdash^{#4}_{#5} \assoctypeS{#6}{#7;#8}}
\newcommand{\WTexpctx}[9]{#2 \joinctx #3 \joinctx #4 \prescript{}{#6}{\vdash^{#5}_{#1}} \assoctype{#7}{#8\joinctx #9}}
\newcommand{\WTvexpr}[4]{#1 \joinctx #2 \vdash \assoctype{#3}{#4}}
\newcommand{\WTvectx}[5]{#2 \joinctx #3 \vdash_{#1} \assoctype{#4}{#5}}
\newcommand{\WTenv}[3]{  #1 \vdash \assocctx{#2}{#3}}
\newcommand{\WTproc}[5]{ #1 \joinctx #2 \vdash^{#3} \assocctx{#4}{#5}}
\newcommand{\WTprocS}[5]{ #1 \joinctx #2 \vdash^{#3} \assocctxS{#4}{#5}}
\newcommand{\WTconf}[6][\aaactorcontext]{#2 \joinctx #3 \vdash^{#4}_{#1} \assocctx{#5}{#6}}
\newcommand{\WTthread}[8]{#1 \joinctx #2 \joinctx #3 \vdash^{#4}_{#5} \assoctype{#6}{#7\joinctx#8}}
\newcommand{\WTlab}[6]{  #1 \joinctx #2 \vdash_{#3} \assocctx{#4}{#5\joinctx#6}}

\newcommand{\refRule}[1]{{\small{\textsc{(#1)}}}\xspace} 
\newcommand{\NamedRule}[4]{\scriptstyle{\refRule{#1}}\
\displaystyle                  
\frac{#2}{#3}         
\begin{array}{l}
#4     
\end{array}
}
\newcommand{\Rule}[3]
{\displaystyle                  
\frac{#1}{#2}         
\begin{array}{l}
#3     
\end{array}
}

\newcommand{\Barista}{\code{B}}
\newcommand{\Customer}{\code{Cs}}
\newcommand{\Counter}{\code{Cn}}
\newcommand{\Coffee}{\code{Cf}}

\newcommand{\CleanCup}{\code{CC}}
\newcommand{\main}{\code{main}}
\newcommand{\order}{\code{order}}
\newcommand{\f}{\code{f}}
\newcommand{\fO}{\code{f1}}
\newcommand{\ffO}{\code{ff1}}

\newcommand{\fF}{\code{f4}}
\newcommand{\ffF}{\code{ff4}}
\newcommand{\x}{\code{x}}

\newcommand{\cT}{\code{cc}}
\newcommand{\cO}{\code{c}}

\newcommand{\unitS}{\code{u}}

\newcommand{\place}{\code{plc}}

\newcommand{\takeorder}{\code{tkOr}}
\newcommand{\oC}{\code{o}}
\newcommand{\makeCoffee}{\code{mkCf}}

\newcommand{\aEnv}{\actorenvironment}
\newcommand{\yO}{\code{y1}}
\newcommand{\yT}{\code{y2}}
\newcommand{\yTh}{\code{y3}}
\newcommand{\yF}{\code{y4}}

\newcommand{\releaseexpS}[2]{\key{rls} \,#1\,#2}
\newcommand{\letinS}[2]{\key{let} \;#1\,\key=\,#2...}
\newcommand{\retS}[1]{\key{ret} \;#1}

\newcommand{\varfut}{vr-ft}
\newcommand{\varres}{vr-rs}
\newcommand{\varother}{vr-oth}

\newcommand{\veval}{val}

\newcommand{\expawait}{e-awt}
\newcommand{\expcall}{e-cl}
\newcommand{\exprls}{e-rls}
\newcommand{\exphold}{e-hld}
\newcommand{\expop}{e-op}
\newcommand{\expchoiceleft}{e-ch-l}
\newcommand{\expchoiceright}{e-ch-r}
\newcommand{\expletred}{e-let-red}
\newcommand{\explet}{e-let}

\newcommand{\cswap}{$\Pcong$-swap}
\newcommand{\cactivate}{$\Pcong$-act}
\newcommand{\cyield}{$\Pcong$-yld}

\newcommand{\procget}{get}
\newcommand{\procrestart}{restart}
\newcommand{\procsuspend}{suspend}
\newcommand{\prochold}{hold}
\newcommand{\procrls}{rls}

\newcommand{\procreturn}{return}
\newcommand{\procspawn}{spawn}
\newcommand{\proccall}{call}
\newcommand{\procsilent}{silent}
\newcommand{\proclpar}{par-l}
\newcommand{\procrpar}{par-r}
\newcommand{\proccong}{p-cong} 

\newcommand{\tyve}{T-ve}
\newcommand{\tyexp}{T-exp}
\newcommand{\tyemptyenv}{T-env-empty}
\newcommand{\tyenv}{T-env}
\newcommand{\tyconf}{T-conf}

\newcommand{\tyvar}{T-var}

\newcommand{\tyval}{T-val}
\newcommand{\tyvalF}{T-val-f}
\newcommand{\tyvalR}{T-val-r}
\newcommand{\tyvarF}{T-var-f}
\newcommand{\tyvarR}{T-var-r}

\newcommand{\tyawait}{T-awt}
\newcommand{\tycall}{T-call}
\newcommand{\tyrls}{T-rls}
\newcommand{\tyhold}{T-hold}
\newcommand{\tyop}{T-op}
\newcommand{\tychoice}{T-ch}
\newcommand{\tylet}{T-let}
\newcommand{\tyreturn}{T-ret}

\newcommand{\tymsg}{T-msg}
\newcommand{\tythread}{T-thrd}
\newcommand{\typarallel}{T-par}
\newcommand{\tyidle}{T-idle}
\newcommand{\tydone}{T-done}
\newcommand{\tyresult}{T-res}

\newcommand{\tylblresult}{T-lab-res}
\newcommand{\tylblcall}{T-lab-call}
\newcommand{\tylblsilent}{T-lab-silent}
\newcommand{\tylblrls}{T-lab-rls}
\newcommand{\tylblhold}{T-lab-hold}

\newcommand{\NN}{\mathsf{Nat}_\infty}
\newcommand{\NNeq}{\NN^{=}}
\newcommand{\NNleq}{\NN^{\leq}}
\newcommand{\N}{\mathbb{N}}
\newcommand{\Lin}{\mathsf{Lin}}
\newcommand{\LinSet}{|\Lin|}
\newcommand{\Lev}{\mathsf{Lev}}
\newcommand{\LevSet}{|\Lev|}
\newcommand{\Priv}{\mathsf{priv}}
\newcommand{\Pub}{\mathsf{pub}}

\newcommand{\code}[1]{\texttt{#1}}

\makeatletter
\newcommand{\onelettername}[1]{#1\aftergroup\@gobble}
\makeatother

\newcommand{\futassoc}[3][]{#2 :\ifblank{#1}{}{_{#1}} #3}
\newcommand{\fmarklab}{\mu}
\newcommand{\fmark}{\mathbf{\chi}}

\newcommand{\emptyctx}{\emptyset}

\newcommand{\ConfS}[2]{#1\,||\,#2}       
\newcommand{\parallelcompS}[2]{#1 {\parP}#2}

\newcommand{\NamedRuleOL}[3]{{\scriptstyle{\refRule{#1}}}\ #2\ #3}

\begin{document}
\title{Fair Termination for Resource-Aware \\Active Objects}
%
%
\author{Francesco Dagnino\inst{1}\orcidlink{0000-0003-3599-3535} \and
Paola Giannini\inst{2}\orcidlink{0000-0003-2239-9529} \and
{Violet Ka I} Pun\inst{3}\orcidlink{0000-0002-8763-5548}\and \\
Ulises Torrella\inst{3}\orcidlink{0009-0009-4430-5774}}
\authorrunning{F. Dagnino et al.}
%
\institute{DIBRIS, Universit\`a di Genova, Italy
 \email{francesco.dagnino@dibris.unige.it}
 \and
DiSSTE, Universit\`a del Piemonte Orientale, Italy\\
\email{paola.giannini@uniupo.it}
\and
Western Norway University of Applied Sciences, Norway\\
\email{\{violet.ka.i.pun,unto\}@hvl.no}}
\maketitle              
\vspace{-10pt}
\begin{abstract}
Active object systems are a model of distributed computation that has been adopted for modelling distributed systems and business process workflows. This field of modelling is, in essence, concurrent and resource-aware, motivating the development of resource-aware formalisations on the active object model. 
The contributions of this work are the development of a core calculus for resource-aware active objects together with a type system ensuring that well-typed programs are fairly terminating, i.e., 
they can always eventually terminate.
To achieve this, we combine techniques from graded semantics and type systems, which are quite well understood for sequential programs, with those for fair termination, which have been developed for synchronous~sessions.

\keywords{
graded types \and
active-objects \and
resource-aware \and
fair termination \and
workflows

}
\end{abstract}
%
%
%

\section{Introduction}

Active object systems~\cite{BoerSHHRDJSKFY17} are the object-oriented instantiation of the \emph{actor model}~\cite{Agha85,BakerH77}. 
They provide a useful abstraction of distributed systems with asynchronous communications, which are represented as collections of active objects (actors)\footnote{We use actors and active objects interchangeably when it is clear from the context.} interacting through asynchronous method calls and futures. 
Active object languages, such as those presented in \cite{SchaferP10,JohnsenHSSS10,BrandauerCCFJPT15},
capture this form of concurrency model by means of 
distributed multithreaded computational entities with cooperative scheduling, where each actor can handle multiple messages at a time,
by explicitly yielding control when awaiting on certain conditions to be fulfilled, for instance, a resolved future.

Active objects have been adopted for 
formalising and analysing workflow models~\cite{AliLP23,EasyRPL} by capturing the behaviour of the internal (resource-sensitive) processes of organisations.
Workflow processes handle business cases (a customer order, a service ticket, etc.) by executing a sequence of tasks, and are primarily demanded to reach case resolution, hence termination~\cite{Aalst97,Aalst98}.
\bfd Active objects directly capture the workflow management notion of \bpa active resources \epa (employees, machines, etc.), 
representing the operational power a system possesses to carry out its tasks.
As such, these are available in a fixed number and cannot be dynamically instantiated. 
However, workflow models often need to take into account also \efd 
\emph{passive/informational resources}~\cite{Muehlen04},\footnote{\bfd  In this paper, ``resources'' always refer to passive/informational resources. \efd \VPComm{another suggestion: In this paper, ``resources'' always refer to passive/informational resources.}}
which, while performing a workflow, 
move through processes, undergo transformations, can be created or destroyed, and crucially, have a \emph{limited availability} according to the specification domain.
The interplay of asynchronous message passing, cooperative multithreading and resource management makes it challenging to ensure that a system can complete its tasks and consequently be terminating.
For instance,  if a thread tries to access a resource that is not available, it remains stuck as it cannot yield control, thus preventing the whole system from successfully terminating.  

To model such systems, we propose a core calculus for resource aware active objects. 
\bfd This calculus models systems with a fixed number of active objects which manipulate dynamic (passive) resources. 
These \efd are represented as values  with a limited availability, 
described by a \emph{grade} annotation constraining their usage. 
For instance, a resource can be indicated for a quantitative use with a numerical grade, or for a mode of use with a \emph{public/private} grade, ensuring that private resources are not made public. 
Resources are local to each object and can only be shared by message passing. 
Moreover, all threads of an actor can hold and release the resources it owns, 
where getting hold of resources is \emph{blocking} if the requested amount is not available.

It is important to note that the introduction of graded resources has an impact on the synchronisation mechanism.
Indeed, typically futures can be accessed an arbitrary number of times \cite{SomayyajulaP22,NiehrenSS06,JohnsenHSSS10}. 
However, this is not the case in our setting because futures may contain graded resources and so, by copying the future, we would copy its content as well, leading to a violation of the constraint on the resource usage expressed by the grade. 
To overcome this issue, we treat futures \emph{linearly}, allowing them to be read only once. 

 We endow our calculus with 
a type system ensuring that well-typed programs are \emph{fairly terminating}~\cite{CicconeDP24}. 
This result does not forbid non-termination, but rather it guarantees that termination is always possible. 
Hence, well-typed programs, unless they systematically avoid going towards termination, which is considered an unfair behaviour, are guaranteed to terminate. 
This result guarantees several desirable properties of concurrent systems, such as {\em livelock freedom} and {\em orphan message freedom}, in addition to 
the main goal of this work, which is that resources are handled correctly, i.e., actors  do not run out of resources they need for executing their tasks. 
In this context livelock freedom means that an object waiting for a future, \bpa the result of an asynchronous method call, \epa will eventually get it and 
orphan message freedom means that a message requiring the execution of a method will eventually be read and the method executed.
To achieve this, we combine techniques from graded semantics and type systems~\cite{BrunelGMZ14,AbelB20,ChoudhuryEEW21,DalLagoG22a,BianchiniDGZ23oopsla,TorczonAAVW24}, which are quite well understood for sequential programs, with those for fair termination~\cite{CicconeP22,CicconeP22concur,CicconeDP24,DagninoP24} that have been developed for synchronous sessions. 
%

The rest of the paper is structured as follows: Section~\ref{sect:grade-monoid} introduces the algebraic preliminaries that are used in the paper and Section~\ref{sec:language} defines our core calculus.
The type system guaranteeing fair termination as well as its soundness are presented in Section~\ref{sec:typesystem}.  Finally, we discuss some related work and conclude the paper with future work in Section~\ref{sec:related.conclusion}.
Omitted rules, results and full proofs of the theorems in  Section~\ref{sec:typesystem} can be found in the Appendix.



\section{Algebraic Preliminaries: Grade Monoids and Subtraction}
\label{sect:grade-monoid}

Resource-aware semantics and type systems are often parameterised over an algebraic structure whose elements, usually dubbed \emph{grades}, model the ``modes'' of availability of resources, e.g., their number of copies, their privacy level or the degree of imprecision affecting them. 
This structure is typically a variant of ordered semirings  
\cite{BrunelGMZ14,GaboardiKOBU16,McBride16,Atkey18,AbelB20,ChoudhuryEEW21,DalLagoG22a,BianchiniDGZ23oopsla,TorczonAAVW24} 
with addition and multiplication operations for combining grades and a compatible order relation enabling their approximation. 
 
In this section, we introduce the algebraic structure of grades we use in this work, which is different from the previously mentioned algebras. 
In particular, multiplication is not needed while \emph{subtraction} is required.
The latter is a novelty of this paper in the current context, thus we describe it in more detail.
%
We start by introducing the notion of grade monoid. 

\begin{definition} \label{def:gr-mon} 
A \emph{grade monoid}  is an ordered commutative monoid 
$\GMon = \ple{\GSet,\gord,\gsum,\gzero}$
such that, for all $\gr\in\GSet$, $\gr\gord\gzero$ implies $\gr = \gzero$. 
\end{definition}

This means that, in a grade monoid $\GMon = \ple{\GSet,\gord,\gsum,\gzero}$, 
$\gsum$ is an associative and commutative binary operation on $\GSet$ with $\gzero$ as a neutral element and 
$\gord$ is a partial order on $\GSet$ making $\gsum$ monotone in both arguments. 
The last requirement ensures that there is no grade below $\gzero$. 
This is reasonable since $\gzero$ intuitively models the absence of a resource. 

\begin{example}\label{ex:gr-mon}
We show some standard examples adapted from the literature, 
e.g.,~\cite{BianchiniDGZ23oopsla}. 
\begin{enumerate}
\item\label{ex:gr-mon:nat}
The ordered monoids 
$\NNeq = \ple{\N+\{\infty\},=,+,0}$ and $\NNleq = \ple{\N+\{\infty\},\leq,+,0}$
of extended natural numbers, ordered either by equality or by the standard ordering on them, with addition and zero, are grade monoids. 
The former tracks \emph{exact usage} of resources, while the latter tracks \emph{bounded usage} of resources. 
\item\label{ex:gr-mon:lin}
The ordered monoid $\Lin = \ple{\LinSet,\le,+,0}$ where 
$\LinSet = \{0,1,\infty\}$, with $0\le \infty $ and $1\le\infty$, and the addition is given by
\vspace{-.75em}
\[ 0 + x = x + 0 = x \qquad 1 + 1 = \infty \qquad \infty + x = x + \infty = \infty \]

\vspace{-.75em}
for all $x\in\LinSet$, is a grade monoid. 
It tracks resources that are either not used ($0$), or used exactly once ($1$) or in an unrestricted way ($\infty$). 
 Adding $0\leq 1$ we get the affinity~grade. 
%
\item\label{ex:gr-mon:lev}
Let $\Lev = \ple{\LevSet,\gord,\lor,0}$ be a join-semilattice. 
Then, $\Lev$ is a grade monoid. 
Intuitively, elements of $\LevSet$ represent the (maximum) mode in which a resource can be used. 
For instance, if $\LevSet = \{0,\Priv,\Pub\}$ is a set of privacy levels with $0\le\Priv\le\Pub$, 
a resource with grade $\Priv$ can only be used in a private mode. 
\end{enumerate}
\end{example}

\noindent
Recall that grades $\gr\in\GSet$ such that $\gzero\gord\gr$ play a special role: 
they represent usages that can be \emph{discarded}, as they can be reduced to $\gzero$ through the approximation relation. 
For instance, in the linearity grade monoid \refItem{ex:gr-mon}{lin}, 
the grade~$\infty$ is discardable, while~$1$ is not, modelling the fact that $1$ represents resources which must be used exactly once.  
 In the affinity grade~$1$ also is discardable as it means usage at most once. 

As mentioned, 
we extend the algebraic structure of grades with a subtraction operation. 
This will be used in the operational semantics to model resource consumption in a deterministic way. 
Hence, we define it to be a partial binary operation, which we expect to be undefined when the current grade is not enough to cover the required consumption.  More precisely, we have the following definition. 

\begin{definition}\label{def:subtract-op}
Let $\GMon = \ple{\GSet,\gord,\gsum,\gzero}$ be a grade monoid. 
A \emph{subtraction} is a partial binary function 
$\pfun{\gsub}{\GSet\times\GSet}{\GSet}$ 
such that, for all $\gr,\ggr,\ggr'\in\GSet$, 
if $\ggr\gsub\gr$ is defined and $\ggr\gord\ggr'$ then 
$\ggr'\gsub\gr$ is defined and $\ggr\gsub\gr\gord\ggr'\gsub\gr$. 
\end{definition}

A subtraction operation is a partial binary function, which is monotonic in the first argument. 
Note that here monotonicity also requires that subtraction is defined on the larger argument. 

\begin{definition} \label{def:subtract-mon}
A \emph{subtractive  grade monoid} is a structure $\GMon = \ple{\GSet,\gord,\gsum,\gzero,\gsub}$, where 
\ple{\GSet,\gord,\gsum,\gzero} is a grade monoid and $\gsub$ is a subtraction operation on it such that, 
for all $\gr,\ggr,\ggr'\in\GSet$, $\gr\gsum\ggr'\gord\ggr$ if and only if $\ggr\gsub\gr$ is defined and $\ggr'\gord\ggr\gsub\gr$. 
\end{definition}

Intuitively, in a subtractive grade monoid the subtraction $\ggr\gsub\gr$ is defined whenever there is a residual grade $\ggr'$ such that $\gr\gsum\ggr'\gord\ggr$. 
When $\ggr\gsub\gr$ is defined, we have 
$\gr\gsum(\ggr\gsub\gr)\gord\ggr$. 
This means that, when it is defined, $\ggr-\gr$ is the largest grade that added to $\gr$ stays below $\ggr$. 
As a consequence, the partial function $\pfun{(\bpa\cdot\epa)\gsub\gr}{\GSet}{\GSet}$, when it is defined, behaves  as a right adjoint of the function $\fun{\gr\gsum(\bpa\cdot\epa)}{\GSet}{\GSet}$. 
Finally, notice that, from \cref{def:subtract-mon}, we get that, if $\gr\gord\ggr$, then 
$\ggr\gsub\gr$ is defined and it is discardable. 
Indeed, we have $\gr \gsum\gzero = \gr \gord \ggr$  and this implies 
$\gzero \gord \ggr\gsub\gr$. 

\begin{example} \label{ex:subtract-mon}
We show how to define subtraction operations on grade monoids of \cref{ex:gr-mon}.  
\begin{enumerate}
\item\label{ex:subtract-mon:nat}
In both grade monoids $\NNeq$ and $\NNleq$ of \refItem{ex:gr-mon}{nat}, 
the subtraction $y - x$ is defined only when $x\le y$  and in this
case  is given by
\vspace{-.75em}
\[ y - x = \sup\{ z \in \N+\{\infty\} \mid x+ z \le y \} \]

\vspace{-.75em}
Note that when $x \le y$ the set of which we take the supremum is not empty as $0$ belongs to it. 
Furthermore, this definition implies that $\infty - x = \infty$ for all $x\in\N+\{\infty\}$.
%
\item\label{ex:subtract-mon:lin}
In the grade monoid $\Lin$ of \refItem{ex:gr-mon}{lin},   the subtraction is defined by the following table
$x - 0 = x$ and $\infty - x = \infty$, for all $x\in\LinSet$, and 
$1 - 1 = 0$. 
 Similarly for the affine grade. 
Note that, even if $1 + 1 = \infty$, we have that $1\ne (1+1)-1 = \infty - 1 = \infty$, that is, even if we have a subtraction, the monoid does not need to have inverses. 
%
\item\label{ex:subtract-mon:lev}
In the grade monoid $\Lev$ of \refItem{ex:gr-mon}{lev}, 
observe that, if $x \le y$, we have $x \lor y = y$ and so $y - x$ must be defined and should be equal to $y$. 
On the other hand, 
if $y-x$ is defined, then $x \lor (y-x) \le y$, which implies that $x \le y$. 
Therefore, we have that $y - x $ is defined if and only if $x \le y$ and in this case we have $y - x = y$. 
\end{enumerate}
\end{example}
Other useful properties of subtractive grade monoids can be found  in \cref{app:background}.
\PAComm{For arxiv: ``Other useful properties of subtractive grade monoids can be found  in \cref{app:background}.''}


\begin{figure}[!b]
\begin{math}
    \begin{array}[h]{l@{\hspace{2ex}}l}
        \begin{array}[t]{rcll}
            \expression & \define &  \callexp{\actorid}{\methid}{\many{\valueexpression}} \mid \awaitexp{\valueexpression} \\
			& \mid & \holdexp{\gradeid}{\resid} \mid \releaseexp{\gradeid}{\valueexpression} \\ 
			& \mid & \opexp{\many{\valueexpression}} \mid \nondetchoic{\expression_1}{\expression_2}  \\
			& \mid & \return{\valueexpression} \mid \letexp{\var}{\expression_1}{\expression_2}\\
            \valueexpression & \define & \gradedvar{\var}{\gradeid} \mid \var \mid \val \\ 
            \val & \define &  \res \mid \futid \mid \unit\\
        \end{array} 
        & 
        \begin{array}[t]{rcll}
            \activationrecord & \define & \many{ \assocval\var\val} \\
            \processid, \pprocessid & \define & \activethread{\activationrecord}{\expression}{\futid}{\actorid} \mid \idleactor{\actorid} \mid \queuedthread{\activationrecord}{\expression}{\futid}{\actorid}  \\ 
                & \mid & \callmsg{\futid}{\actorid}{\methid}{\many {\val}} \mid \placefut{\futid}{\val} \mid \parallelcomp{\processid}{\pprocessid} \mid \pdone\\ 
            \resenv & \define & \many\res \\ 
           \actorcontext, \aactorcontext, \aaactorcontext & \define & \many{\actorid: \resenv} &\\
            \config & \define & \Conf{\actorenvironment}{\processid}  &\\
        \end{array} 
    \end{array} 
  \end{math}
\setlength\belowcaptionskip{-15pt}
\caption{Syntax where \bpa $\gr$ and $\resid$ are grades and resource identifiers, respectively\epa}
\label{fig:syntax} \label{fig:runtimesyntax}
\end{figure}

\section{A Calculus of Resource-Aware Active Objects}
\label{sec:language}
In this section we introduce the syntax and operational semantics of
our core calculus, which are based on active object languages~\cite{BoerSHHRDJSKFY17}.
It models systems~with a fixed number of objects communicating through asynchronous method calls and synchronising over futures. 
Active objects are multithreaded, thus, 
they can handle multiple messages at a time, but with a cooperative scheduling, i.e., only one active thread is allowed and it has to explicitly yield control to the others. 

The distinctive feature of our calculus is that it models \emph{resources} needed by objects for carrying out their tasks. 
Resources are located within objects and have a \emph{limited availability}. 
Each thread of an object can hold and release a certain amount of a resource, and crucially, is not allowed to hold a resource that is not available in the object.
Therefore, we abstractly represent resources as constants decorated by a \emph{grade}, which specifies their availability. 
In the following we fix a subtractive grade monoid $\GMon = \ple{\GSet,\gord,\gsum,\gzero,\gsub}$.
The syntax of the language is shown in \cref{fig:syntax}. 
With $\actorid$ and $\methid$ we denote actor and method names, with $\var, \vvar$ variable names, 
 and $\resid$ and $\gradeid$ are respectively resource names and grades. Overbar
 denotes a list of elements, e.g., $\many{\ve}$ stands for $\ve_1,\dots \ve_n$ for $n\geq 0$.
Actors have associated methods and resources. 
The methods are abstractly modelled by the function
\vspace{-.75em}
\begin{quoting}
  \centering
\begin{math} 
  \mbody{\actorid}{\methid} = \MBody{\many\var}{\e}
\end{math}
\end{quoting}
%
which gives the parameters and body of method~$\methid$ of actor~$\actorid$. 
Similarly, resources owned by each actor are specified by a mapping,  dubbed {\em actor context}, denoted as~$\actorcontext$,~$\aactorcontext$ and~$\aaactorcontext$,  
associating with each actor $\actorid$ its resource environment~$\resenv$, which is a sequence of graded resources with unique names, 
i.e., it is a finite partial function from resources to their grades. 

The syntax of {\em expressions} is fine-grained~\cite{LevyPT03} and it is parameterised over a signature $\Sigop$ of primitive operations, ranged over by $\key{op}$, for creating, transforming and destroying resources. 
{\em Value expressions} $\ve$ are either (graded) variables or values, which can be 
graded resources, futures or $\unit$. 
Expression $ \callexp{\actorid}{\methid}{\many{\valueexpression}}$ is the request of the asynchronous execution of the method $\methid$ 
of actor $\actorid$ on parameters $\many {\ve}$. This expression will return a (new) future, that can be  used to wait for the result of the call through the expression $\awaitexp{\valueexpression}$. 
Resources can be added to or taken from the resource environment of the actor running the method by \key{release} and \key{hold}
expressions, specifying the resource name and the amount (the grade) to be transferred. 
The expression $\opexp{\many\ve}$ denotes a call to the primitive operation $\key{op}$ of $\Sigop$. 
The construct  $\nondetchoic{\e_1}{\e_2}$ 
models a non-deterministic choice between the execution of $\e_1$ and $\e_2$. 
Expression $\return{\valueexpression}$ ends the execution and returns
the value of $\ve$,
and finally the \key{let} construct binds the variable $\var$ to the value of the expression~$\e_1$
in the evaluation of~$\e_2$. 
\begin{example}\label{ex:barista}

\begin{figure}[b]
  \centering
  \begin{subfigure}{0.47\linewidth}
\begin{lstlisting}
(*\label{lst:sig-op-init}*)makeCoffee(x: Order(*$^1$*), y: CleanCup(*$^1$*))
          : Coffee(*$^1$*)
washCup(x: DirtyCup(*$^1$*)): CleanCup(*$^1$*)
(*\label{lst:sig-op-end}*)drink(x: Coffee(*$^1$*)): DirtyCup(*$^1$*)

Customer {
 main(): Unit {
(*\label{lst:sc.c.order}*)  let f1 = Barista!takeOrder(order)in 
(*\label{lst:sc.c.finishOrder}*)  let x = f1? in
(*\label{lst:sc.c.gotCoffee}*)  let f2 = Counter!pickup() in
(*\label{lst:sc.c.coffeeInHand}*)  let x1 = f2? in
(*\label{lst:sc.c.drinkCoffee}*)  let x2 = drink(x1(*$^1$*)) in 
(*\label{lst:sc.c.clean}*)  let f3 = Barista!clean(x2(*$^1$*)) in f3?;
(*\label{lst:sc.c.ndchoice.1}*)  (return unit
   (*$\ch$*)
(*\label{lst:sc.c.ndchoice.2}*)   y=Customer!main();y?;return unit } 
}
\end{lstlisting}
  \end{subfigure}
\hfill
  \begin{subfigure}{0.46\linewidth}
\begin{lstlisting}[firstnumber=18]
Counter {  
 place(c : Coffee): Unit {
(*\label{lst:sc.ct.placeCoffee}*)  release 1 c(*$^1$*);return unit }
 pickup(): Coffee(*$^1$*) {
(*\label{lst:sc.ct.getCoffee}*)  let x=hold 1 Coffee in return x(*$^1$*) }
}

Barista {
 takeOrder (o : Order(*$^1$*)): Unit {
(*\label{lst:sc.b.getCup}*)  let cc = hold 1 CleanCup in 
(*\label{lst:sc.b.makeCoffee}*)  let c = makeCoffee(o(*$^1$*),cc(*$^1$*)) in
(*\label{lst:sc.b.putCoffee}*)  let f4 = Counter!place(c(*$^1$*)) in 
(*\label{lst:sc.b.takeOrder.return}*)  f4?; return unit } 
 clean(dc : DirtyCup(*$^1$*)): Unit {
(*\label{lst:sc.b.washCup}*)  let cc = washCup(dc(*$^1$*)) in
(*\label{lst:sc.t.releaseCup}*)  release 1 cc;return unit } 
}
\end{lstlisting}
\end{subfigure}
\vspace{-10pt}
  \setlength\belowcaptionskip{-15pt}
\caption{\label{fig:ex.simplecafe}A workflow example of a cafe}  
\end{figure}
\cref{fig:ex.simplecafe} shows a simple workflow of a cafe modelled in our calculus, involving three active objects $\code{Customer}$, $\code{Barista}$ and $\code{Counter}$, and three types of resources, namely, clean cups, dirty cups and coffee.
For simplicity, we use~``;'' instead of let-constructs for sequential composition wherever it is obvious.
The customer starts the workflow by (asynchronously) ordering a coffee from the barista (Line~\ref{lst:sc.c.order}) and then waits.
The latter, upon receiving the order, grabs a clean cup (Line~\ref{lst:sc.b.getCup}) to make a coffee (Line~\ref{lst:sc.b.makeCoffee}).
After that, the barista places it 
on the counter (Line~\ref{lst:sc.b.putCoffee}),
and notifies the customer 
(Line~\ref{lst:sc.b.takeOrder.return}). 
The awaiting customer now picks up the coffee from the counter and drinks it (Lines~\mbox{\ref{lst:sc.c.gotCoffee}--\ref{lst:sc.c.drinkCoffee}}), which subsequently produces a dirty cup, 
which s/he then asks the barista to clean (Line~\ref{lst:sc.c.clean}).
Upon request, the barista washes the dirty cup and produces a clean cup which is now available for the next coffee (Line~\ref{lst:sc.t.releaseCup}).
When the cup is cleaned, the customer can non-deterministically choose
to leave the cafe or order another coffee by repeating the workflow
(Lines~\ref{lst:sc.c.ndchoice.1}--\ref{lst:sc.c.ndchoice.2}).

The manipulation of resources is captured by the operations
$\code{makeCoffee}$, $\code{drink}$  and $\code{washCup}$
\bpa (whose signatures are given in Lines~\ref{lst:sig-op-init}--\ref{lst:sig-op-end} and used in  Lines~\ref{lst:sc.c.drinkCoffee}, \ref{lst:sc.b.makeCoffee}~and~\ref{lst:sc.b.washCup}) \epa as well as
the \key{hold} and \key{release} statements on
Lines~\ref{lst:sc.ct.placeCoffee},~\ref{lst:sc.ct.getCoffee},~\ref{lst:sc.b.getCup},
and~\ref{lst:sc.t.releaseCup}.
The barista uses $\code{makeCoffee}$, which takes an order and a clean cup, to produce a coffee, and adopts $\code{washCup}$ to convert a dirty cup into a clean one.
%
Getting hold of resources, \emph{one} coffee and \emph{one} clean cup, is modelled by Lines~\ref{lst:sc.ct.getCoffee} and~\ref{lst:sc.b.getCup}, respectively, while the release of these two resources is captured by Lines~\ref{lst:sc.ct.placeCoffee} and~\ref{lst:sc.t.releaseCup}.
Note that the resource quantity to be held or released is explicitly specified, and 
the action of holding resources is only possible if the number of available resources can fulfil the requested amount.
Observe that the program is \emph{resource safe} 
only if the barista starts with at least one clean cup.
With the non-deterministic choice in Lines~\ref{lst:sc.c.ndchoice.1}--\ref{lst:sc.c.ndchoice.2}, non-termination is possible in this workflow model, and it is fairly terminating if it is resource safe.
\end{example}

{\em System configurations~$\config$} consist of the association between actors and their resources~$\actorenvironment$, and a process~$\processid$, which is 
a parallel composition of threads and messages. 
A \emph{running} actor comprises precisely one {\em active thread} $\activethread{\activationrecord}{\expression}{\futid}{\actorid}$ and
any number of {\em suspended threads} $\queuedthread{\activationrecord}{\expression}{\futid}{\actorid}$,
where~$\activationrecord$ is a local environment mapping variables to values,~$\expression$ an expression to be executed and~$\futid$ a future that hold the final result.  
An actor can also be \emph{idle}, signaled by the presence of $\idleactor{\actorid}$. 

There are two kinds of messages: 
$\callmsg{\futid}{\actorid}{\methid}{\many {\val}}$ denotes a call to method~$\methid$ of actor~$\actorid$ with parameters $\many\val$ and expects the final result on future~$\futid$; 
conversely, $\placefut\futid\val$ represents a completed computation
whose value~$\val$ is stored in future~$\futid$. 
Threads and messages are composed by the
 operator $\parP$ which is assumed to be associative and to have $\pdone$ as neutral element.
However, $\parP$ is \emph{not} 
commutative as futures impose a dependency between thread and messages, preventing us from freely swapping them, which will be clarified in the definition of the structural precongruence (see \cref{fig:processeval}). 

In the following we assume that configurations are well formed, meaning that 
for any actor appearing in a configuration, it either has exactly one active thread or it is idle.
This assumption ensures that each multithreaded actors, i.e., active objects, can only execute one task at a time.
Our operational semantics presented below enforces that one thread can
pass control to another only in specific situations.
See \cref{ex:reduction} for examples of configurations. 
%

We now introduce the  {\em operational semantics} of our active object based calculus  by presenting the most significant rules.
The rest of the rules are given in \cref{app:language}.
We first define the semantics of (value) expressions. 
This is a resource-aware semantics in the style of \cite{ChoudhuryEEW21,BianchiniDGZ23ecoop,BianchiniDGZ23oopsla,TorczonAAVW24}. 
Notably, the evaluation takes place within a local environment and, 
instead of performing substitutions, variables are replaced one at a time when needed, consuming their corresponding value in the local environment. 
This behaviour is realised by the relation 
$\veeval\activationrecord\ve \vestep \veeval{\activationrecord'}{\val}$,
which reduces a value expression to a value, possibly modifying the local environment. 
The main rule is
\begin{center}
$
\NamedRule{\varres}{}
{ 
   \veeval{\activationrecord, \assocval{\var}{\gradedr{\resid}{\ggradeid}}
 }
  {\gradedvar{\var}{\gradeid}}
            \vestep
                \veeval{
                    \activationrecord, \assocval{\var}{\gradedr{\resid}{\ggradeid \gsub \gradeid}}
                }
                {\gradedr{\resid}{\gradeid}}
} 
{}
$
\end{center}
\noindent
which reduces a graded variable, with grade~$\gradeid$, to a graded resource $\gradedr\resid\gradeid$, provided that the variable is associated with $\gradedr\resid\ggradeid$ in the local environment and the subtraction $\ggradeid\gsub\gradeid$ is defined; 
the local environment is then updated, mapping the variable to the resource $\gradedr\resid{\ggradeid\gsub\gradeid}$. 
In this way we model resource consumption: 
the reduction is \textit{stuck} if the current amount~$\ggradeid$ of the resource is not enough to satisfy the required amount~$\gradeid$ (i.e., $\ggradeid\gsub\gradeid$ is undefined); otherwise, the amount~$\gradeid$ is consumed, leaving only  the difference. 
Compared to the literature, here resource consumption is \emph{deterministic} thanks to the use of the subtraction operation.
This allows us to prove a standard subject reduction theorem
(\cref{lem:srve} in \cref{apx:proof}) which in general does not hold for resource-aware semantics.

\begin{figure}[t]
\begin{footnotesize}
\begin{math}
\begin{array}{c}
\NamedRule{\expcall}
                {
                    \veeval{\activationrecord_i}{\valueexpression_i}
                    \vestep
                    \veeval{\activationrecord_{i+1}}{\val_i}\quad \forall i\in 1..n
                }
                {
                    \expeval{\activationrecord_1}{\callexp{\actorid}{\methid}{\many{\valueexpression}}}
                    \expstep{\callmsg{\futid}{\actorid}{\methid}{\many\val}}
                    \expeval{\activationrecord_{n+1}}{\return{\futid}}
                }
                {
                 \fresh(\futid)
                 \ \ 
                  }
 \NamedRule{\expawait}
                        {
                            \veeval{\activationrecord}{\valueexpression}
                            \vestep
                            \veeval{\activationrecord'}{\futid}
                        }
                        {
                            \expeval{\activationrecord}{\awaitexp{\valueexpression}}
                            \expstep[\placefut{\futid}{\val}]{}
                            \expeval{\activationrecord'}{\return{\val}}
                        } {}  
 \\[4ex]
 \NamedRule{\exprls}      
                    {
                        \veeval{\activationrecord}{\valueexpression}
                        \vestep
                        \veeval{\activationrecord'}{\gradedr{\resid}{\ggradeid}}
                    }
                    {
                        \expeval{\activationrecord}{\releaseexp{\gradeid}{\valueexpression}}
                        \expstep{\rlsmsg{\gradedr{\resid}{\gradeid}}}
                        \expeval{\activationrecord'}{\return{\unit}}
                    }

 {\gradeid{\gord}\ggradeid}\ \ 
\NamedRule{\exphold}
                   {}
                    {
                        \expeval{\activationrecord}{\holdexp{\gradeid}{\resid}}
                        \expstep[\holdmsg{{\gradedr{\resid}{\gradeid}}}]{}
                        \expeval{\activationrecord}{\return{\gradedr{\resid}{\gradeid}}}
                    }{}
\\[4ex]
\bpa\NamedRule{\explet}
                        {\veeval{\activationrecord}{\valueexpression}
                            \vestep
                            \veeval{\activationrecord'}{\val}
                        }
                        {
                            \expeval{\activationrecord}{\letexp{\var}{{\return{\ve}}}{\expression_2}}
                            \expstep{\emptylabel}
                            \expeval{\activationrecord',\assocval{\vvar}{\val}}{\Subst{\expression_2}{\var}{\vvar}}
                        }
          { \fresh(\vvar)                }\epa
          \end{array}
\end{math}
\end{footnotesize}
\setlength\belowcaptionskip{-15pt}
\caption{Evaluation of  expressions}    \label{fig:expeval}\label{fig:rulesve}
\end{figure}

\paragraph{Evaluation of expressions.}
The semantics of expressions is described by a labelled reduction relation 
$\expeval\activationrecord\expression \expstep[\labelid_i]{\labelid_o} \expeval{\activationrecord'}{\expression'}$, 
whose most representative rules are given in  \cref{fig:expeval}. 
Labels record how an expression interacts with the external environment, 
and are defined as follows:
\begin{quoting}
  \centering
\begin{math} 
 \labelid  \define  \holdmsg{\gradedr{\resid}{\gradeid}} \mid \rlsmsg{\gradedr{\resid}{\gradeid}} \mid \callmsg{\futid}{\actorid}{\methid}{\many{\val}} \mid 
            \placefut{\futid}{\val} \mid  
            \emptylabel
\end{math}
\end{quoting}
The first two labels state how the resources are taken from or released to the current
actor context: 
by $\holdmsg\res$ the expression requests from the environment the resource $\resid$ with grade $\gradeid$, 
while by $\rlsmsg\res$ it moves to the environment the resource $\resid$ with grade $\gradeid$. 
Communications between actors are enabled by the two following labels,
$\callmsg{\futid}{\actorid}{\methid}{\many{\val}}$ and
$\placefut{\futid}{\val}$, which are interpreted as the two kinds of
messages of the same shape in the configurations described above. 
Finally, the label $\emptylabel$ denotes the absence of interaction. 

Labels are divided into two classes:
$\rlsmsg\res$ and $\callmsg{\futid}{\actorid}{\methid}{\many{\val}} $
refer to outputs that the expression produces  towards the environment, while 
$\holdmsg\res$ and $\placefut\futid\val$ represent inputs that the expression reads from it. 
To highlight this semantic difference, 
we write labels of the former class above the arrow denoting the reduction relation, while those of the latter class will appear below it. 
The label~$\emptylabel$ can appear in both places, but we will often omit it for clarity. 

Rule \refRule\expcall evaluates the parameters, produces the corresponding label and returns a fresh future. 
Note that here we do not use the $\aux{mbody}$ function, i.e., we do not access the method table. 
In this way, this expression is never stuck modelling the fact that calling a method corresponds to sending a message which will be asynchronously read. 
Rule \refRule\expawait evaluates the value expression to a future and returns the corresponding value reported in the label. 
Rule \refRule\exprls evaluates the value expression to a graded resource with grade~$\ggradeid$, 
and checks if~$\ggradeid$ suffices to cover the amount to be released and, in this case, it produces the corresponding label and returns~$\unit$.
Rule \refRule\exphold returns the graded resource reported in the label.
\bpa Rule \refRule\explet adds the local variable $\var$  to the environment, modulo renaming with a fresh variable to avoid clashes. 
No interaction with the external environment is involved. \epa
%

\begin{figure}[!t]
\begin{small}
\begin{math}
\begin{array}{ll}
            \refRule\cswap 
&
            \processid \parP \pprocessid 
   \Cong 
            \pprocessid \parP \processid 
\qquad\qquad
            \text{if}\quad \fp\processid\cap\fr\pprocessid = \emptyset\text{ and } \fp\pprocessid\cap\fr\processid=\emptyset 
            \\[1ex]
            \refRule\cactivate
            &  
            \parallelcomp{\idleactor{\actorid}}{
                \queuedthread{\activationrecord}{\expression}{\futid}{\actorid}
            }
            \Pcong 
            \activethread{\activationrecord}{\expression}{\futid}{\actorid}
            \\[1ex]
            \refRule\cyield
            &  
            \activethread{\activationrecord}{\expression}{\futid}{\actorid}
           \Pcong 
            \parallelcomp{\idleactor{\actorid}}{
                \queuedthread{\activationrecord}{\expression}{\futid}{\actorid}
            } 
\qquad\qquad
            \text{if}\quad \expeval\activationrecord\expression \expstep[\placefut{\futid'}\val]{}\\[2ex]
\end{array} 
\end{math}
\hrule\vspace{1ex} 

\begin{math}
\begin{array}{c}
 \NamedRule{\proccall}
            {
                \expeval{\activationrecord}{\expression}
                \expstep{\callmsg{\futid}{\actorid}{\methid}{\many{\val}}}
                \expeval{\activationrecord'}{\expression'}
            }
            {
                \Conf{\actorenvironment}
                {
                    \activethread{\activationrecord}{\expression}{\futid'}{\aactorid}
                }
                \semanticstep
                \Conf{\actorenvironment}
                {
                    \parallelcomp{\callmsg{\futid}{\actorid}{\methid}{\many{\val}}}{
                        \activethread{\activationrecord'}{\expression'}{\futid'}{\aactorid}
                    }
                }
            }{}
\\[4ex]
\NamedRule{\procspawn}
            {}
            {
                \Conf{
                    {\actorenvironment}
                }{
                  \parallelcomp{\callmsg{\futid}{\actorid}{\methid}{\many{\val}}}{\idleactor{\actorid}}
                }
                
                \semanticstep
               \Conf{\actorenvironment}{\activethread{\many{\var\mapsto\val}}{\expression}{\futid}{\actorid}}
            }
            {
             \quad \mbody{\actorid}{\methid} = \MBody{\many\var}{\expression} 
            }
  \\[2ex]
\NamedRule{\procget}
            {
                \expeval{\activationrecord}{\expression}
                \expstep[\placefut{\futid}{\val}]{}
                \expeval{\activationrecord'}{\expression'}
            }
            {
                \Conf{\actorenvironment}{
                \parallelcomp{
                    \placefut{\futid}{\val}}{
                        \activethread{\activationrecord}{\expression}{\futid'}{\actorid}}}
                \semanticstep 
                \Conf{
                    \actorenvironment
                }{
                    \activethread{\activationrecord'}{\expression'}{\futid'}{\actorid}
                }
            }{}
 \\[3ex]
\NamedRule{\procrls}
            {
                \expeval{\activationrecord}{\expression}
                \expstep{\rlsmsg{\gradedr{\resid}{\gradeid}}}
                \expeval{\activationrecord'}{\expression'}
            }
            {
                \Conf
                {
                    \actorenvironment, a: \resenv, \gradedr{\resid}{\ggradeid}
                }
                {
                    \activethread{\activationrecord}{\expression}{\futid}{\actorid}
                }
                \semanticstep
                \Conf
                {
                    \actorenvironment, a: \resenv, \gradedr{\resid}{\ggradeid + \gradeid}
                }
                {
                    \activethread{\activationrecord'}{\expression'}{\futid}{\actorid}
                }
            }{}
\\[3ex]
\NamedRule{\prochold}
            {
                \expeval{\activationrecord}{\expression}
                \expstep[\holdmsg{\gradedr{\resid}{\gradeid}}]{}
                \expeval{\activationrecord'}{\expression'}
            }
            {
                \Conf
                {
                    \actorenvironment, \actorid: \resenv, \gradedr{\resid}{\ggradeid}
                }
                {
                    \activethread{\activationrecord}{\expression}{\futid}{\actorid}
                }
                \semanticstep
                \Conf
                {
                    \actorenvironment, \actorid: \resenv, \gradedr{\resid}{\ggradeid \gsub \gradeid}
                }
                {
                    \activethread{\activationrecord'}{\expression'}{\futid}{\actorid}
                }
            }{}
\end{array}
\end{math}
\end{small}
\setlength\belowcaptionskip{-10pt}
\caption{The precongruence on processes and the reduction on configurations.}    \label{fig:processeval}
\end{figure}

\paragraph{Precongruence and Reduction on configurations.}
The main semantic relation is a reduction 
$\config\red\config'$ on configurations, whose rules are given in the bottom section of \cref{fig:processeval}. 
As it is customary in process calculi, the definition of the reduction relies on a precongruence relation on processes $\processid\Pcong\pprocessid$, generated by the clauses in the top section of \cref{fig:processeval}. 
We write $\processid\Cong\pprocessid$ when both $\processid\Pcong\pprocessid$ and $\pprocessid\Pcong\processid$ hold. 
This relation accounts for bureaucratic rearrangements of threads and messages within a process. 
As mentioned, 
the $\parP$ operator is assumed to be associative, but it is not commutative.
The precongruence relation specifies in which case we can swap two processes in a parallel composition. 
To this end, we consider, for every process $\processid$, the sets 
$\fp\processid$ and $\fr\processid$ of futures \emph{produced} and \emph{consumed}, respectively, which, intuitively, represent the input and output  ``channels'' through which $\processid$ can be connected to other processes.
For a thread $\fp\processid$  is the future specified as  superscript and for a message is the
target of the arrow. For a thread $\fr\processid$  is the set of futures occurring in the expression in its body plus the ones in the  values of the local environment and
for a message is  the set of futures occurring in values which are parameters or result of the message.
For parallel composition we take the union of the set of the  futures consumed by the process on the left and the ones
of the process on the right from which we first remove the future produced by the process on the left, which are not consumed
by the  parallel composition. The formal definition can be found in \cref{app:language}.

%
Let us now consider the rules for $\Pcong$.  Clause \refRule\cswap states that we can swap two processes only if they are disconnected, i.e., they do not depend on each other. 
The other two clauses instead regulate the (de)activation of threads of a certain actor:
\refRule\cactivate enables the activation of any thread of an actor, provided that the actor is idle, 
while \refRule\cyield states that an actor can suspend the execution
of the current thread to become idle only if it is awaiting on a future. 
This last rule makes the relation $\Pcong$ only a precongruence and not a congruence: 
the clause \refRule\cactivate can be inverted only if the thread is awaiting on a future (by \refRule\cyield),  in all other cases the execution of a thread cannot be suspended. 

%

Rules \refRule\proccall, \refRule\procspawn and \refRule\procget realize the asynchronous communication mechanism of our calculus. 
When an active thread calls a method, Rule \refRule\proccall adds the call to the parallel composition on the \emph{left} of the active thread.
If an idle actor has a pending call to itself, it can start a new thread executing the method body by Rule \refRule\procspawn. Note that if the actor does not have the called method, i.e., $\mbody\actorid\methid$ is undefined, this rule is not applicable and the call will never be executed. 
If an active thread is awaiting on a fulfilled future,  
it can retrieve the value from the future and remove the latter from the parallel composition (Rule \refRule\procget). 
Observe that both the new pending call in Rule \refRule\proccall and the fulfilled future in Rule \refRule\procget are on the \emph{left} of the active thread, so that the structure of produced and consumed futures is respected. 
Resource management is modelled by Rules \refRule\procrls and \refRule\prochold. 
The former allows an active thread to add to the actor's resource environment resources at a certain grade, while the latter allows an active thread to move resources from the actor's resource environment to its local environment at a certain grade. 
Note that Rule \refRule\prochold is applicable only if the actor resource environment has enough resources to satisfy the request, i.e., the subtraction $\ggr\gsub\gr$ is defined. 
%
The precongruence on processes and the reduction are  propagated to parallel composition in the obvious way. 

\begin{example}\label{ex:reduction}
In \cref{fig:red} we show a \textit{possible} reduction for \cref{ex:barista} starting from an initial configuration,~$\config$, containing a message for the $\code{Customer}$ asking to execute
method $\code{main}$, when all the actors are idle. Moreover, the $\code{Barista}$ has a $\code{CleanCup}$,  i.e.,

\begin{quoting}
\begin{math}
\config=\ConfS{\aEnv}{\parallelcompS{\callmsg{\f}{\Customer}{\main}{}}{\parallelcompS{ \idleactor{\Barista}}{\parallelcompS{ \idleactor{\Customer}}{ \idleactor{\Counter}}}}}
\quad\text{  where  }\quad\aEnv={\Barista}:{\gradedr\CleanCup{1}}
\end{math}
\end{quoting}
where names are shortened taking their initials, among other natural abbreviations.
In the reductions we underline the processes that are reduced modulo precongruence.

In Line 1, the $\code{Barista}$, who is idle, answers to the pending request of the execution of  $\code{takeOrder}$ by starting an active thread containing the body of $\code{takeOrder}$.
The configuration in Line 2 is produced by the evaluation of expression $\code{hold}$ of a $\code{CleanCup}$ by the running thread of the $\code{Barista}$, which subtract a 
$\code{CleanCup}$ from his resources and returns it to the expression, whose evaluation, in the next step, places it in the local environment of the thread. 
The reduction in Line 3 shows how the evaluation of a method call consumes the resources on $\activationrecord_1$ required for its arguments, creates a new future $\ffF$ and adds a message with the call 
to the configuration, and finally reduces the call expression to $\ffF$.
The message is added to the left of the thread, so the future $\ffF$,  once fulfilled, can be  consumed by the thread. 
The configuration in Line 4 is produced by the evaluation of expression $\code{release}$ of a $\code{Coffee}$
by the $\code{Counter}$, which puts it in its resource environment  removing it  from its local environment, so the actor context becomes $\actorcontext_1$. 
In this example all the actors were either idle with no queued threads or running a thread. 
However, we could apply Rule \refRule\cyield of the pre-congruence  to
the active thread of $\code{Custumer}$ with  any reduction starting from Line 2 to produce  $\small{\parallelcomp{ {\queuedthread{\assocval\yO\ffO}{\letinS{\x}{\yO?}}{\f}{\Customer}} }{{ \idleactor{\Customer}}}}$ since
 $\small{\expeval{\assocval\yO\ffO}{\letinS{\x}{\yO?}} \expstep[\placefut{\ffO}\val]{}\expeval{\;}{\letinS{\x}{\retS \ffO}}}$.

\begin{figure}[t]
\begin{footnotesize}
\begin{math}
\begin{array}{l}
{\hspace{-1.5ex}}\config \red\Conf{\aEnv}{{\parallelcomp{ 
\underline{ \activethread{\;}{\letinS{\fO}{ \callexp{\Barista}{\takeorder}{\order}}}{\f}{\Customer}}  }{\parallelcomp{ \idleactor{\Barista}}{ \idleactor{\Counter}}}}}
             \\[.7ex]       
\begin{array}{ll}
{\red}&\Conf{\aEnv}{\parallelcomp{\underline{\callmsg{\ffO}{\Barista}{\takeorder}{\order} }} 
{\parallelcomp{ \activethread{\;}{\letinS{\fO}{ \retS\ffO} }{\f}{\Customer} }{\parallelcomp{ \underline{\idleactor{\Barista}}}{ \idleactor{\Counter}}}}}\\[.7ex]
{\hspace{-3ex}1. \red}&\Conf{\aEnv}{\parallelcomp{{\activethread{ \assocval\oC\order}{\letinS{\cT}{ \holdexp{\code{1}}{\Coffee}}}{\ffO}{\Barista} }}}
{\parallelcomp{ \underline{\activethread{\;}{\letinS{\fO}{\retS\ffO}}{\f}{\Customer}} }{{ \idleactor{\Counter}}}}\\[.7ex]
{\red}&\Conf{\aEnv}{\parallelcomp{
\underline{\activethread{ \assocval\oC\order}{\letinS{\cT}{ \holdexp{\code{1}}{\CleanCup}}}{\ffO}{\Barista} }}}
   {\parallelcomp{ \activethread{  \activationrecord}{\letinS{\x}{\yO?}}{\f}{\Customer} }{{ \idleactor{\Counter}}}}\\[.7ex]
{\hspace{-3ex}2. \red}&\Conf{\aEnv_0}
          {\parallelcomp {\underline{\activethread{ \assocval\oC\order}{\letinS{\cT}{ \retS\gradedr{\CleanCup}{1} } }{\ffO}{\Barista} }}
                                 {\parallelcomp{ \activethread
{  \activationrecord}{...}{\f}{\Customer} }{{ \idleactor{\Counter}}}}
                                 }\\[.7ex]
{\red}&\Conf {
                  \aEnv_0
                      }
                   { \parallelcomp 
                                     {\underline{\activethread{ \activationrecord_0 } {\letinS{\cO}{ \makeCoffee(\oC,\yT)  } } {\ffO}{\Barista}}}                          
                              {\parallelcomp{ \activethread
{  \activationrecord}{...}{\f}{\Customer} }{{ \idleactor{\Counter}}}}
                                 }
                    \\[.7ex]
{\red}&\Conf {
                \aEnv_0
                      }
                   { \parallelcomp 
                                     {\underline{\activethread{ \activationrecord_0 } {\letinS{\cO}{\retS{\gradedr\Coffee {1} }} } {\ffO}{\Barista} }}                          
                              {\parallelcomp{ \activethread
{  \activationrecord}{...}{\f}{\Customer} }{{ \idleactor{\Counter}}}}
                                 }
                    \\[.7ex]
{\red}&\Conf {
                  \aEnv_0
                      }
                   { \parallelcomp 
                                     {\underline{\activethread{ \activationrecord_1} {\letinS{\fF}{ \callexp{\Counter}{\place}{\yTh}  } } {\ffO}{\Barista}} }                          
                              {\parallelcomp{ \activethread
{  \activationrecord}{...}{\f}{\Customer} }{{ \idleactor{\Counter}}}}
                                 }
                    \\[.7ex]
{\hspace{-3ex}3. \red}&\Conf {
                  \aEnv_0
                      }
                   { \parallelcomp{\underline{\callmsg{\ffF}{\Counter}{\place}{\gradedr{\Coffee}{1}}}}
                            { \parallelcomp 
                                     {\activethread
{ \activationrecord_2} {\letinS{\fF}{ \retS\ffF } } {\ffO}{\Barista} }                          
                              {\parallelcomp{ \activethread
{  \activationrecord}{...}{\f}{\Customer} }{{ \underline{\idleactor{\Counter} } }}}
                              }
                         }
                    \\[.7ex]
{\red}&\Conf {
                  \aEnv_0
                      }
                   { \parallelcomp  {\underline{\activethread{ \assocval\cO{\gradedr{\Coffee}{1}} } {\releaseexpS{1}{\gradedr\cO{1}};..} {\ffF}{\Counter} }}
                            { \parallelcomp 
                                     {\activethread
{ \activationrecord_2} {\letinS{\fF}{ \retS\ffF } } {\ffO}{\Barista} }                          
                              { \activethread
{  \activationrecord}{...}{\f}{\Customer} }
                              }
                         }
                    \\[.7ex]
{\hspace{-3ex}4.\red}&\Conf {
                  \aEnv_1
                      }
                   { \parallelcomp  {\underline{\activethread{ \assocval\cO{\gradedr{\Coffee}{0}} } {{\retS\unitS}} {\ffF}{\Counter}} }
                            { \parallelcomp 
                                     {\activethread
{ \activationrecord_2} {\letinS{\fF}{ \retS\ffF } } {\ffO}{\Barista} }                          
                              { \activethread
{  \activationrecord}{...}{\f}{\Customer} }
                              }
                         }
                    \\[.7ex]
{\red}&\Conf {  \aEnv_1
                      }
                   { \parallelcomp  { \parallelcomp{\placefut{\ffF}{\unitS}}{ \idleactor{\Counter}} }
                            { \parallelcomp 
                                     {\underline{\activethread{ \activationrecord_2} {\letinS{\fF}{ \retS\ffF } } {\ffO}{\Barista} }}                          
                            { \activethread
{  \activationrecord}{...}{\f}{\Customer} }
                              }
                         }
                    \\[.7ex]
{\red}&\Conf {
                  \aEnv_1
                      }
                   { \parallelcomp  { \parallelcomp{\underline{\placefut{\ffF}{\unitS}}}{ \idleactor{\Counter}} }
                            { \parallelcomp 
                                     {\underline{\activethread{ \activationrecord_3} {\yF?; ..} {\ffO}{\Barista} }  }                        
                             { \activethread
{  \activationrecord}{...}{\f}{\Customer} }
                              }
                         }
                    \\[.7ex]
{\red}&\Conf {
                  \aEnv_1
                      }
                   { \parallelcomp { \idleactor{\Counter}} 
                            { \parallelcomp 
                                     {\activethread{ \activationrecord_3} {\retS\unitS} {\ffO}{\Barista} }                          
                             { \activethread
{  \activationrecord}{...}{\f}{\Customer} }
                              }
                         }
\quad {\red} \qquad ... 
                          \\[1ex]
 \end{array} \\[20ex]
  \end{array}\\
  \begin{array}{l}
\aEnv_0={\Barista}:{\gradedr\CleanCup{0}}\quad \aEnv_1={\Barista}:{\gradedr\CleanCup{0}},\Counter:\gradedr{\Coffee}{1}\qquad\quad
 \activationrecord= \assocval\yO\ffO \quad
\activationrecord_0=\assocval\oC\order,\assocval\yT{\gradedr{\CleanCup}{1} }\\
 \activationrecord_1=\activationrecord_0,\assocval\yTh{\gradedr{\Coffee}{1} }\quad
 \activationrecord_2=\activationrecord_0,\assocval\yTh{\gradedr{\Coffee}{0} }
\quad \activationrecord_3=\activationrecord_2,\assocval\yF{\ffF }
 \end{array}
\end{math}
\end{footnotesize}
\vspace{-.5em}
\setlength\belowcaptionskip{-15pt}
\caption{A reduction sequence for \cref{ex:barista}}\label{fig:red}
\end{figure}
\end{example}

\section{The Type System}
\label{sec:typesystem}

In this section we introduce a type system ensuring that well-typed configurations can always successfully complete their execution. 
To this end, beside the standard concurrency-related errors such as deadlocks, 
we also have to prevent incorrect usages of resources, 
i.e., whenever a thread tries to access a resource, it must be available with the requested amount in the resource environment of the actor running the thread. 
To achieve this, we adapt techniques from graded type systems to track resource requirements in typing judgments of processes and expressions. 
The key challenge here is that the availability of resources crucially depends on when threads are executed because, for instance, one may rely on a resource released by another one. 
Thus, the tracking of resource requirements and the future-based synchronisation mechanism should properly interact to ensure a correct resource usage.
\paragraph{Types and type contexts.} Types are then defined by the grammar:
\vspace{-.5em}
\[
\T \define \Unit \mid \gradedr\resid\gradeid \mid \FutTy{\T}{\actorcontext}
\]

\vspace{-.5em}
\noindent
where $\Unit$ is the type of $\unit$, 
$\gradedr\resid\gr$ is the type of a resource with name $\resid$ that has to be used according to the grade $\gr$, and 
$\FutTy{\T}{\actorcontext}$ is the type of a future. 
This specifies, beside the type $\T$ of the held value, 
an actor context $\actorcontext$ saying which resources will be available after synchronising on a future of this type.
Note that for simplicity resource types are just resource names, i.e., they are singleton types. 


Typing judgments rely,  in addition to the actor contexts introduced in \cref{sec:language},  on  standard variable contexts 
$\typingcontext,\ttypingcontext \define \many{\var:\T}$, 
which are finite partial maps  from variables to types,  and 
on future contexts 
$\futcontext,\ffutcontext \define \many{\futassoc[\fmarklab]\futid{\FutTy{\T}{\actorcontext}}}$, 
which are also finite partial maps from future names to their types and marking~$\fmarklab$ that is either empty or $\fmark$. 
The use of the marking will be clarified later in this section. 

As it is customary in graded type systems, we need operations for combining contexts. 
In particular, we consider the following operations and relations. 
\begin{itemize}
\item We define addition and partial order on actor contexts $\actorcontext\gsum\aactorcontext$ and $\actorcontext\gord\actorcontext$, 
by extending pointwise the addition and the partial order of the grade monoid~$\GMon$ first to resource environments~$\resenv$ 
and then to actor contexts, assuming that resources and actors not appearing in the contexts associated with~$\gzero$ and an empty resource environment, respectively. 
%
\item The sum of typing contexts  $\typingcontext=\typingcontext_1\tysum\typingcontext_2$ is defined
by constructing the union of the variables defined in one context and not in the other and 
\begin{quoting}
\begin{math}
\begin{array}{l}
\typingcontext_1(\var)=\T_1\ \wedge\ \typingcontext_2(\var)=\T_2\implies \\
\quad\quad(\T_1=\T_2=\Unit\ \wedge\ \typingcontext(\var)=\Unit)\ \vee\\
\quad\quad(\T_1=\gradedr{\resid}{\gradeid_1}\ \wedge\ \T_2=\gradedr{\resid}{\gradeid_2}\ \wedge\ \typingcontext(\var)=
\gradedr{\resid}{\gradeid_1\gsum\gradeid_2})
\end{array}
\end{math}
\end{quoting}
As we can see a variable with a future type can be defined either in $\typingcontext_1$ or $\typingcontext_2$ but not in both, i.e., futures are typed linearly.
%
\item for all contexts we will use a comma to denote their disjoint union. 
For instance, $\futcontext,\ffutcontext$ is the union of $\futcontext$ and $\ffutcontext$, which is defined only if they have disjoint domains. 
\end{itemize}
Finally, we say that the typing context {\em $\typingcontext$ is discardable}, dubbed $ \dis{\typingcontext}$, if no variable has a future type and
for variables having a resource type, $\res$,  the grade~$\gr$ must be ``discardable'', i.e., $\gzero\gord\gr$. This reflects the fact that futures are treated linearly and
that our grade monoid is not necessarily affine.


\paragraph{Typing rules for expressions.} The typing of value expressions is given by the judgement,
$\WTvexpr\futcontext{\typingcontext}{\ve}\T$ 
where $\typingcontext$ and $\futcontext$ contains variables and futures used by the value expression. The most significant rule is 
\[
  \NamedRule{\tyvarR}
  {}
  {
     \WTvexpr{\emptycontext}{\var:\gradedr\resid\ggradeid,\typingcontext}{\gradedr\var\gradeid}\res
   }  
  {\quad\gradeid \gord \ggradeid\ \wedge\  \dis{\typingcontext}}
\]
\noindent
that types a graded variable $\gradedr\var\gr$ provided that $\var$
appears in the context with a grade $\ggr$ larger than or equal to $\gr$. 
Note that in this rule (as in the others (except  for
\refRule\tyvalF) in \cref{fig:types4expr:app}) \UTComm{remove in APLAS}the future context is empty  and additional (not used) variables are allowed only
provided that they are discardable. 
The complete set of typing rules
for value expressions can be found in the upper part of
\cref{fig:types4expr:app} in \cref{app:typesystem}.
 \PAComm{For arxiv: ``The complete set of typing rules
for value expressions can be found in the upper part of
\cref{fig:types4expr:app} in \cref{app:typesystem}.''}
%


%
\begin{figure}[t]
 \begin{small}
\begin{math}
\begin{array}{c}
\NamedRule{\tycall}
  {
     \WTvexpr{\futcontext_i}{\typingcontext_i }{\ve_i}{\T_i}\quad i\in 1\dots n
  }
  {
    \WTexpr{\actorcontext\actsum\rescontext}{ \futcontext}{\typingcontext}{k+3}{\actorid}{\callexp{\aactorid}{\methid}{\many{\valueexpression}}}{  \FutTy{\T}{\aactorcontext}}{\rescontext}
  } 
  {
   \futcontext = \futcontext_1 \linearsum \dots \linearsum \futcontext_n\\
  \typingcontext = \typingcontext_1 \tysum \cdots \tysum \typingcontext_n \\
    \mtype{\aactorid}{\methid} =\MType{\actorcontext}{\many\T}{\T}{k}{\aactorcontext}}
 \\[4ex]
\NamedRule{\tyawait}
  { 
    \WTvexpr{\futcontext}{\typingcontext }{\valueexpression}{\FutTy{\T}{\aactorcontext}}
  }
   { 
   \WTexpr{\actorcontext}{ \futcontext}{\typingcontext}{1}{\actorid}{ \awaitexp{\ve}}{\T}{\actorcontext\actsum\aactorcontext}
  }
  {}          
\NamedRule{\tyrls}
    {
      \WTvexpr{\futcontext}{\typingcontext}{\ve}{\gradedr{\resid}{\ggradeid}}
    }
    {
       \WTexprS{\actorcontext}{ \futcontext}{\typingcontext}{1}{\actorid}{\releaseexp{\gradeid}{\ve}}{\Unit}{\actorid{:}\res\actsum\actorcontext}
    }
   {\gradeid {\gord} \ggradeid} 
    \\[3ex]
\NamedRuleOL{\tyhold}
  {
    \WTexprS{\actorid:\rres\actsum\actorcontext}{\emptycontext}{\emptycontext}{1}{\actorid}{ \holdexp{\gradeid}{\resid} }{\res}{\actorcontext}
  }
  {\quad\gradeid {\gord} \ggradeid}
 \\[2ex]
  \NamedRule{\tylet}
     {
         \WTexpr{\actorcontext_1}{ \futcontext_1}{\typingcontext_1}{m}{\actorid}{\e_1}{\T'}{\actorcontext_1'\actsum\actorcontext_2'}\quad\quad
         \WTexpr{\actorcontext_2\actsum\actorcontext_1'}{ \futcontext_2}{\typingcontext_2,\var:\T'}{n}{\actorid}{\e_2}{\T}{\aactorcontext_2}
  }
  {
   \WTexpr{\actorcontext_1\actsum\actorcontext_2}{ \futcontext_1,\futcontext_2}{\typingcontext_1\tysum\typingcontext_2}{1+n+m}{\actorid}
                {\letexp{\var}{\expression_1}{\expression_2}}{\T}{\actorcontext_2'\actsum\aactorcontext_2}
   }
   {}
\end{array}    
\end{math}
\end{small}
\setlength\belowcaptionskip{-15pt}
\caption{Typing rules for expressions}
\label{fig:types4expr}
\end{figure}

Expressions are typed through the judgement 
$\WTexpr{\actorcontext}{\futcontext}{\typingcontext}{n}{\actorid}{\expression}{\T}{\aactorcontext}$,
{partially} defined in \cref{fig:types4expr}. 
Besides the type of the expression and the typing and future contexts containing variables and futures used in it, 
the judgements reports an actor name $\actorid$, a natural number $n$ and two actor contexts $\actorcontext$ and $\aactorcontext$. 
The actor name $\actorid$  records the {\em actor running this expression}. 
The context $\actorcontext$, dubbed {\em required actor context},  keeps track of the resources the expression requires from each actor for being executed, 
while $\aactorcontext$, dubbed {\em produced actor context},  those it releases to the system, which could be used by subsequent expressions or by other threads having access  to the future which this expression writes its result to. 
Finally, the number $n$ provides a \emph{measure} in number of steps of the effort required to terminate the execution of the expression. 

To type method and primitive operation calls, we assume  some additional typing information.
For methods, apart from the information about parameters and body, given by $\mbody{\actorid}{\methid}$,  we assume to have their typing information
\[
\mtype{\actorid}{\methid}=\MType{\actorcontext}{\many\T}{\T}{n}{\aactorcontext}
\]

\noindent
Furthermore, we assume that when this is defined, then
\[
\mbody{\actorid}{\methid} = \MBody{\many\var}{\expression} \text{ and } 
\WTexpr{\actorcontext}{\emptycontext}{\many{\var:\T} }{n}{\actorid}{\expression}{\T}{\aactorcontext}
\]
%
that is, $\mtype\actorid\methid$ and $\mbody\actorid\methid$ are consistent. 

Rule \refRule\tycall checks that the arguments of a method call  have the types specified by $\aux{mtype}$,
and assigns to the call a future type containing the return type of the method and the actor context it promises to make available when its execution terminates. 
The resources required by the method are added to the actor context required by the call, while the produced one is not affected. 
The measure is obtained by adding $3$ to the measure of the method. 
These $3$ additional units account for the three steps needed: one for sending the message, one for starting the thread handling the call, and one for terminating it. 

Rule \refRule\tyawait checks that the value expression has a future type, 
assigns to the await expression the type of the value held by the future and 
adds to the actor environment produced by the await the actor context reported in the future type. 
In this way, we allow resources promised by the future to be used. 
The required actor context instead is not affected. 
Rules \refRule\tyhold and \refRule\tyrls govern the holding and releasing of resources. 
 A \key{hold} expression has a graded resource type with the same name and grade of the requested resource 
and adds it to the required actor context for the actor running the expression, but with a larger grade.   
Conversely, the argument of a \key{release} expression must have a graded resource type with a grade  
larger than the one that should be released. 
The resource is added with the released grade to the produced actor context for the actor running the expression. 
%
%
Rule \refRule\tylet is pretty standard except for the handling of actor contexts. 
The required and produced actor contexts of the \key{let} expression are obtained by summing up those of the two subexpressions, but removing from both the context $\actorcontext_1'$. 
Indeed, this context provides resources which are used \emph{internally} by the \key{let} expression: 
they are produced by $\e_1$ and consumed by $\e_2$. 
Therefore, we do not have to track them in the final typing. 
The measure of the \key{let} is just~$1$ plus the measure of the two subexpressions, so that we account for the step needed for passing the result of $\e_1$ to~$\e_2$. 
%
%

The variable and the future contexts in the conclusion of these rules are obtained by combining those in the premises: 
variable contexts are summed up with $\tysum$, while future contexts
are combined using the comma, 
ensuring that variables of graded resource types have a grade covering the needs of all premises and those with a future types and futures themselves are used only in one of~the premises. 
Finally, all rules except \refRule\tylet
have an actor context added to both the required and the produced
ones in the conclusion, thus allowing an expression to promise additional resources provided that it requires them to the environment. 
This makes the rules more flexible, simplifying the proof of subject reduction. 
The remaining typing rules for expressions can be found
  in the central part of
\cref{fig:types4expr:app} in \cref{app:typesystem}.
 \PAComm{For arxiv: ``The remaining typing rules for expressions can be found
  in the central part of
\cref{fig:types4expr:app} in \cref{app:typesystem}.''}
\begin{figure}[t]
  \hspace{-.5cm}
\begin{math}
  \begin{array}{c}
 \NamedRule{\tythread}     
   { 
        \WTthread\actorcontext{\futcontext}\typingcontext{n}\actorid{\expeval{\activationrecord}{\e}}\T\actorcontext'
  }
  {
      \WTprocS\actorcontext{\futcontext}{n+1}{\anythread[\alpha]{\activationrecord}{\expression}{\futid}{\actorid} }{\futid{:} \FutTy{\T}{\actorcontext'}}
  }{ \alpha \in \{\ath,\qth\} }
\NamedRule{\tyresult}
   { \WTvexpr{\futcontext}{\emptycontext}{\val}{\T}}
   {
      \WTprocS\actorcontext{\futcontext}{0}{\placefut{\futid}{\val} }{\futid{:} \FutTy{\T}{\actorcontext}}
    } {}
  \\[3ex] 
 \NamedRule{\tymsg}
 {
     \WTvexpr{\futcontext_i}{\emptycontext}{\val}{\T_i}\quad i\in 1\dots n
  }
 {
     \WTprocS\actorcontext{\futcontext_1 \linearsum \dots \linearsum \futcontext_n}{n+2}{ \callmsg{\futid}{\actorid}{\methid}{\many\val}}{\futid{:} \FutTy{\T}{\actorcontext'}}
   }
   {
     \mtype{\actorid}{\methid} =\MType{\actorcontext}{\many\T}{\T}{n}{\actorcontext'}
   }
  \\[3ex]
 \NamedRule{\typarallel}
  {
      \WTproc{\actorcontext_1}{\futcontext_1}{n}{\PP}{\futcontext_1'\linearsum\futcontext''_1}\quad
      \WTproc{\actorcontext_2}{\futcontext_2\linearsum\futcontext'_1}{m}{\PQ}{\futcontext_2'}
  }
  {
    \WTproc{\actorcontext_1\actsum \actorcontext_2}{\futcontext_1\linearsum\futcontext_2}{n+m}{\parallelcomp{\PP}{\PQ}}{\futcontext_2'\linearsum\markcontext(\futcontext_1')\linearsum\futcontext_1''} 
  }
  {
    \dom(\futcontext_1'') \cap \dom(\futcontext_2) = \emptycontext
  } \\[3ex]
    \NamedRule{\tyconf}
  {
  \WTproc{\actorcontext}{\futcontext}{n}{\PP}{\futcontext'}
   }
  {
   \WTconf{\actorcontext_{\rescontext}}\futcontext{n}{ \Conf{\aactorcontext}{\PP}}\futcontext'
  }
  {
    \actorcontext \actsum \rescontext \actord \aactorcontext
  }       
\end{array}
\end{math}
\setlength\belowcaptionskip{-15pt}
\caption{Rules for processes}
 \label{fig:types4proc}
\end{figure}

\paragraph{Typing rules for processes and configurations.}  Processes are typed by the judgement 
$\WTproc\actorcontext\futcontext{n}\processid\ffutcontext$, 
defined in \cref{fig:types4proc}. 
The actor context $\actorcontext$ tracks the resources required by the process. 
The future contexts $\futcontext$ and $\ffutcontext$ records the futures consumed and produced by the process, respectively, with their types. 
That is, one can prove that $\fp\processid = \dom(\futcontext)$ and $\fr\processid = \dom(\ffutcontext)$. 
Note that there is no actor context on the right hand side, because resources can be made available to other thread only through futures.
Rule \refRule\tythread states that a thread (active or suspended) is well-typed
with the actor context required by the running expression. 
The required futures are those of the expression in the local environment, while it produces only the future $\futid$ assigned to it. 
Note that the type of the future $\futid$ states that 
it will hold a value of the same type as the running expression and that it will provide resources of the produced actor context of the running expression. 
The measure is $1$ plus the measure of the expression, where $1$ accounts for the step needed for terminating the thread. 
This rule relies on an auxiliary typing judgment for expressions in a local environment
$\WTthread\actorcontext\futcontext\typingcontext{n}\actorid{\expeval{\activationrecord}{\e}}\T\ffutcontext$,  which 
 extracts a variable context from the local environment $\activationrecord$ through the judgment 
 $\WTenv\futcontext\activationrecord{\typingcontext}$ 
and uses it to typecheck the expression~$\e$.   
The resulting actor contexts are those of $\e$, while the future context is the union of those typing $\activationrecord$ and $\e$. 
The formal definition can be found in \cref{app:typesystem}.
\VPComm{For arxiv: ``The formal definition can be found in \cref{app:typesystem}.''}%
Rules \refRule\tyresult and \refRule\tymsg are similar except that they have to check a value and a method call, respectively, instead of an expression in a local environment. 
They only differ for the measure: 
the former has measure~$0$ as it is terminated, while the latter~$2$ plus the measure of the method, where the~$2$ accounts for the steps for starting and terminating the execution of this method. 
%

The parallel composition of two threads is handled by Rule \refRule\typarallel, 
which has to enforce the correct usage of futures, that is, they are  consumed from left to right and used linearly. 
The key fact is that the rule allows the right process $\pprocessid$ to consume futures produced by the left process $\processid$, recorded in the context $\futcontext_1'$.
The other futures used by the two processes, which must be distinct, are propagated in the required future context of the parallel composition. 
The produced future context instead is the union of all the futures produced by the two processes, which again must be distinct. 
Note that the futures in $\futcontext_1'$, i.e., those produced by $\processid$ and consumed by $\pprocessid$, 
are not erased from the final context of produced futures, but rather they are marked by $\fmark$. 
This is needed for remembering that these futures have already been used in the process, and so  they cannot appear elsewhere as they must be used linearly. 
Indeed, it is easy to see that marked futures cannot appear on the left-hand side of a typing judgment, hence they cannot be consumed by other processes. 
Last but not least, the side condition ensures that 
$\futcontext_1'$ contains all futures consumed by $\pprocessid$ matching a future produced by $\processid$.
Also, the actor context and the measure are the sum of those of $\processid$ and $\pprocessid$. 
%

 Lastly, configurations are type-checked  by the judgement 
$\WTconf\actorcontext\futcontext{\actorid}\config\ffutcontext$ and
defined by the Rule \refRule\tyconf. 
This typechecks the process and requires that the actor context of the configuration suffices to cover the one required by the process plus a residual that annotates the judgement. 
The latter takes into account that the process may be part of a larger parallel composition and so the actor context in the configuration  must provide resources also for the rest of the composition, modelled by the residual context.



\paragraph{Soundness.}
%
The main property of the type system is that well-typed configurations are \emph{fairly terminating}~\cite{CicconeDP24}. 
This property  does not forbid infinite computations, but  guarantees that  termination is always possible. 
Moreover, it entails several good properties of concurrent systems with asynchronous communication, such as livelock and orphan message freedom. 

In our context, a configuration is \emph{terminated} if 
it consists only of threads of the form $\idleactor{\actorid}$ or $\placefut{\futid}{\val}$. 
Furthermore, a configuration~$\config$ is said to be \emph{weakly terminating} if 
it can reach a terminated configuration. 
Then, a configuration~$\config$ is \emph{fairly terminating} if, 
whenever $\config\semanticstep\config'$, $\config'$ is weakly terminating, i.e.,~$\config$ can only reach weakly terminating configurations. 


Following the proof strategy of~\cite{CicconeDP24}, fair termination is an immediate consequence of  two key properties of a type system: subject reduction and weak termination. 
In our resource-aware setting, the former is stated as follows: 

\begin{theorem}[Subject reduction]\label{thm:subred}
  If 
  $\WTconf{\actorcontext}{\futcontext}{n}{\config}{\futcontext'} $
  and 
  $ \config \semanticstep \config'$,
  then 
  it exists $\aactorcontext, m, \ffutcontext'$ such that $\WTconf{\aactorcontext}{\futcontext}{m}{\config'}{\ffutcontext'}$,
  where $\ffutcontext'=\futcontext',\markcontext(\ffutcontext'')$ with $\ffutcontext''$ fresh. 
\end{theorem}

This guarantees that, for a well-typed configuration, reduction steps do not affect its ``interface'', i.e., 
the futures consumed and produced by the configuration remains essentially the same, except that the latter ones can grow, but only by fresh and marked ones. 
Furthermore, the residual actor context is also preserved, ensuring that resources required by other processes, that will be composed with the current configuration, are not affected. 

On the other hand, the weak termination result states that well-typed and closed configurations are weakly terminating. 

\begin{theorem}[Weak termination] \label{thm:weak}
  If $
  \WTconf[]{\actorcontext}{\emptycontext}{n}{\config}{\futcontext'}$ then 
  there exists a reduction $\config \semanticstep^{*} \config'$ such that~$\config'$ is terminated.
\end{theorem}

This result relies on a key lemma, dubbed ``helpful direction'', showing that for a well-typed and closed configuration with a positive measure, it is always possible to find a reduction step making the measure decrease. 
Moreover, we have that well-typed and closed configurations with  measure~$0$ are terminated. 
Notice that this auxiliary result encapsulates deadlock freedom. 

Then, since well-typed and closed configurations are weakly terminating and, by subject reduction, they can only reduce to other well-typed and closed configurations, we have that these are fairly terminating. 

%
\bfd Finally, we observe that fair termination implies that well-typed and closed configurations enjoy the following properties: 
\begin{itemize}
\item {\em input-lock freedom}: every future is eventually fulfilled, 
\item {\em orphan message freedom}: every method call is eventually executed, 
\item {\em resource safety}: the execution of \key{hold} expressions is never blocked. 
\end{itemize}
These hold because a well-typed configuration can always eventually terminate, i.e., reach a configuration with no pending threads  and no pending messages, hence 
all futures have been fulfilled,
all method calls have been executed and 
all \key{hold} expressions have succeeded.  \efd
The full proof of fair termination can be found in \cref{apx:proof}.
\PAComm{For arxiv: ``The full proof of fair termination can be found in \cref{apx:proof}.''}
\FDComm{here we should put a reference to arxiv} \PAComm{It is said at the end of the introduction}
\VPComm{added a line to refer to arxiv, can be removed if it is not needed.}


\section{Related Work and Conclusions}
\label{sec:related.conclusion}

We introduced a core calculus for resource-aware active objects, 
which is parametric over a subtractive grade monoid, representing the availability of resources. 
On top of the calculus, we defined a (graded) type system ensuring that well-typed configurations are fairly terminating. 
This property does not forbid non-termination, but ensures that termination is always possible. 
In this way, we can model processes with a potentially infinite execution, 
e.g., a barista taking an arbitrary amount of orders, 
but still retaining good properties such as deadlock, livelock and orphan message freedom. 
In addition, fair termination entails that well-typed configurations are resource-safe, 
i.e., threads always succeed in getting hold of resources. 

\smallskip
\noindent


\paragraph{Related Work}.
Our work is inspired by the modelling of workflow systems with active objects proposed in~\cite{AliLP23}, for the looking at workflows in the active object environment; to 
the resource-aware calculus of~\cite{BianchiniDGZ23oopsla}, \bpa for the association of resources to variables; and to~\cite{CicconeDP24} for 
the use of indexed type system to ensure fair termination.  \epa
%
Resource sensitive active objects have been adopted to model distributed systems and concurrent business process workflows~\cite{JohnsenST15,AliLP23}, by extension of 
core ABS~\cite{JohnsenHSSS10} with explicit notion of time advancement and resources to capture resource consumption. 
In contrast to our core calculus, resources in these two languages are directly accessible by (groups of) active objects, and cannot be transformed.
In addition, resources in~\cite{JohnsenST15} are recovered for each interval when time advances, i.e., they are unlimited.
Graded type systems, also dubbed coeffect systems, 
were introduced in~\cite{PetricekOM13,BrunelGMZ14,PetricekOM14} as a tool for statically controlling how programs use their context, i.e., their free variables, which abstractly represent external resources needed by programs. 
In order to properly state a resource-aware variant of type soundness  for graded type systems, 
graded operational semantics were also considered \cite{ChoudhuryEEW21,BianchiniDGZ23ecoop,BianchiniDGZ23oopsla,TorczonAAVW24}, where variables are annotated by grades and, whenever a variable is used, its grade  is consumed. 
This mechanism is usually non-deterministic, while our calculus supports a deterministic resource consumption thanks to the use of a subtraction operation. 

Finally, in~\cite{Marshall22,MarshallO24} graded types were introduced into the concurrent scenario for a channel-based communication language with in a session-type system. 
There, grades (re)introduce a safe non-linear usage of channels, which are the resources considered in these works. 
We focus on future-based communication of multithreaded active objects, where resources are limitedly available values.

Type systems ensuring fair termination have been studied in the context of synchronous session types \cite{CicconeP22,CicconeDP24,DagninoP24}. 
The key idea to enforce this property, which is the same we have used in our context, is to assign a measure to each process quantifying the effort, i.e., the number of steps, needed to reach termination. 
Indeed, this approach resembles techniques used in the complexity analysis of concurrent systems. 
Notably, \emph{resource-aware session types}\footnote{Note that these types are resource aware in a different sense with respect to our system: in a sense the resource they control is time while we focus on values with a limited availability.}  \cite{DasHP18,DasEtAl21,DasP22}, augment standard session types with measure annotation  designed to perform ammortized cost analysis of concurrent system, where the cost is determined by the number of exchanged messages. 
This approach is closely related to fair termination and a more precise comparison can be found in \cite{DagninoP24}. 
Finally, we mention that in the literature there are other techniques to guarantee termination of concurrent system. 
For instance, Deng and Sangiorgi \cite{DengS04} annotate types with numbers representing priority levels used 
to constrain the order in which channels can be used. 
Note that this technique is designed to  enforce strong normalization, 
i.e., no infinite execution is allowed, which is a much more restrictive property than fair termination, which is the target of our type system.


\smallskip
\noindent
\paragraph{Future work.}
\bfd This work lays a theoretical foundation for the formalisation of complex workflow systems with passive resources. 
To test the effectiveness of the proposed approach, it will be important to develop case studies comparing our system with other solutions and identifying the different subtractive grade monoids required by the concrete application domains. 

On the technical side, we are working on improving some limitations of the system. 
First, \efd 
the linear treatment of futures is a sound but quite restrictive  solution as, for instance, 
it forces to await on futures even if their value is discardable. 
A natural possibility to make the system more flexible is to grade futures as well. 
This poses non-trivial challenges for both the semantics and the type system, and it may require to further extend the algebraic structure of grades. 
Another limitation of the current type system is the way it guarantees fair termination. 
Essentially, it allows programs where termination only depends on the exchanged messages (method calls) and not on the exchanged values (parameters). 
Taking values into account as well would require the adoption of rather sophisticated techniques such as dependent types or sized types with arithmetic refinements \cite{DasP20,DasP22,SomayyajulaP22,BaillotG22}.  
Finally, another interesting extension of this work is to consider the full object-oriented power of active objects, by considering also object instantiation, thus enabling the modelling of systems  with a variable \bfd and dynamically adjusted  \efd number of actors. 
As for the implementation of the type checking algorithm, a promising direction is along the lines of the one in~\cite{OrchardLE19}.


\paragraph*{\textbf{Acknowledgement}.}
{
We would like to thank the anonymous reviewers
for their valuable feedback and constructive suggestions.
This work is part of the \textsc{CroFlow} project:
Enabling Highly Automated Cross-Organisational Workflow Planning,
funded by the Research Council of Norway (grant no. 326249);
and is partially funded by 
\bpa  the MUR project: \textsc{T-LADIES} (PRIN 2020TL3X8X); the COST action:  \textsc{EuroProofNet} (CA20111); 
and has the financial support of the Universit\`{a}  del Piemonte Orientale. \epa
\bfd The first author has been funded by the European Union - NextGenerationEU and
by the Ministry of University and Research (MUR), National Recovery and
Resilience Plan (NRRP), Mission 4, Component 2, Investment 1.5, project ``RAISE
- Robotics and AI for Socio-economic Empowerment'' (ECS00000035). \efd 
}

\bibliographystyle{splncs04}
\bibliography{bib}

\clearpage
\appendix

\section{Definitions and Results Omitted from \cref{sect:grade-monoid,sec:language,sec:typesystem}}
\label{app:sections}


\subsection{\cref{sect:grade-monoid}}
\label{app:background}

The following proposition shows two equivalent conditions that  will be the key property for defining subtractive grade monoids. 

\begin{proposition}\label{prop:subtract-op}
Let $\GMon = \ple{\GSet,\gord,\gsum,\gzero}$ be a grade monoid and $\pfun{\gsub}{\GSet\times\GSet}{\GSet}$ a subtraction operation on $\GMon$. 
Then, the following are equivalent:
\begin{enumerate}
\item\label{prop:subtract-op:1}
for all $\gr,\ggr\in\GSet$, $(\gr\gsum\ggr)\gsub\gr$ is defined and $\ggr\gord(\gr\gsum\ggr)\gsub\gr$ and, 
if $\ggr\gsub\gr$ is defined, then $\gr\gsum (\ggr\gsub\gr)\gord\ggr$; 
\item\label{prop:subtract-op:2}
for all $\gr,\ggr,\ggr' \in \GSet$, 
$\gr\gsum\ggr'\gord\ggr$ if and only if $\ggr\gsub\gr$ is defined and $\ggr'\gord\ggr\gsub\gr$. 
\end{enumerate}
\end{proposition}
\begin{proof}
Let us first prove that \cref{prop:subtract-op:1} implies \cref{prop:subtract-op:2}. 
Suppose that $\gr\gsum\ggr'\gord\ggr$. 
Then, by \cref{prop:subtract-op:1} and \cref{def:subtract-op}, we derive 
$\ggr' \gord (\gr\gsum\ggr')\gsub\gr \gord \ggr\gsub\gr$, as needed. 
Suppose now that $\ggr\gsub\gr$ is defined and $\ggr'\gord\ggr\gsub\gr$. 
Then, by \cref{prop:subtract-op:1}, we derive 
$\gr\gsum\ggr' \gord \gr\gsum(\ggr\gsub\gr) \gord\ggr$, as needed. 

We now show that \cref{prop:subtract-op:2} implies \cref{prop:subtract-op:1}. 
By \cref{prop:subtract-op:2}, we have that 
$\gr\gsum\ggr\gord\gr\gsum\ggr$ implies $\ggr\gord(\gr\gsum\ggr)\gsub\gr$ and , 
when $\ggr\gsub\gr$ is defined, $\ggr\gsub\gr\gord\ggr\gsub\gr$ implies $\gr\gsum(\ggr\gsub\gr)\gord\ggr$, 
as required. 
\end{proof}

Some simple and useful properties of subtractive grade monoids are give by the following proposition.
\begin{proposition} \label{prop:subtract-mon}
Let $\GMon = \ple{\GSet,\gord,\gsum,\gzero,\gsub}$ be a subtractive grade monoid. 
The following facts hold:
\begin{enumerate}
\item\label{prop:subtract-mon:1}
for all $\gr\in\GSet$, $\gzero\gord\gr\gsub\gr$ and $\gr\gsub\gzero = \gr$; 
\item\label{prop:subtract-mon:2}
for all $\gr,\gr'\ggr\in\GSet$, 
if $\ggr\gsub\gr$ is defined and $\gr'\gord\gr$ then 
$\ggr\gsub\gr'$ is defined and $\ggr\gsub\gr\gord\ggr\gsub\gr'$; 
\item\label{prop:subtract-mon:3}
for all $\gr_1,\gr_2,\ggr\in\GSet$, 
$\ggr\gsub(\gr_1\gsum\gr_2) = (\ggr\gsub\gr_1)\gsub\gr_2$;
\item\label{prop:subtract-mon:4}
for all $\gr,\ggr\in\GSet$, 
$\gr\gsum\ggr = \gzero$ if and only if $\gzero\gsub\gr = \ggr$. 
\end{enumerate}
\end{proposition}
\begin{proof}
\cref{prop:subtract-mon:1}. 
From $\gr\gsum\gzero\gord\gr$ we derive 
$\gzero\gord\gr\gsub\gr$ and $\gr\gord\gr\gsub\gzero$. 
Finally, we have $\gr\gsub\gzero = (\gr\gsub\gzero)\gsum\gzero \gord \gr$, which proves the thesis. 

\cref{prop:subtract-mon:2}. 
We have that 
$(\ggr\gsub\gr)\gsum\gr' \gord (\ggr\gsub\gr)\gsum\gr\gord \ggr$ and this implies that 
$\ggr\gsub\gr \gord \ggr\gsub\gr'$, as needed. 

\cref{prop:subtract-mon:3}. 
Suppose that $\ggr\gsub(\gr_1\gsum\gr_2)$ is defined, the proof assuming that $(\ggr\gsub\gr_1)\gsub\gr_2$ is defined is similar. 
From $(\ggr\gsub(\gr_1\gsum\gr_2))\gsum \gr_1\gsum\gr_2\gord\gr$ we deduce that 
$ (\ggr\gsub\gr_1)\gsub\gr_2 $ is defined  and 
$\ggr\gsub(\gr_1\gsum\gr_2) \gord (\ggr\gsub\gr_1)\gsub\gr_2$. 
In addition, we have 
$\gr_1\gsum\gr_2\gsum((\ggr\gsub\gr_1)\gsub\gr_2) \gord \gr_1\gsum(\ggr\gsub\gr_1) \gord \ggr$ 
and this implies 
$(\ggr\gsub\gr_1)\gsub\gr_2\gord\ggr\gsub(\gr_1\gsum\gr_2)$. 

\cref{prop:subtract-mon:4}. 
If $\gzero\gsub\gr = \ggr$, then we have 
$\gr\gsum\ggr\gord\gzero$ and this implies 
$\gr\gsum\ggr = \gzero$ by \cref{def:gr-mon}. 
On the other hand, 
if $\gr\gsum\ggr = \gzero$, then we have 
$\ggr \gord \gzero\gsub\gr 
      = (\gzero\gsub\gr)\gsum\gzero 
      = (\gzero\gsub\gr)\gsum\gr\gsum\ggr 
      \gord \gzero\gsum\ggr 
      = \ggr
$, as needed. 
\end{proof}


\subsection{\cref{sec:language}}
\label{app:language}

\begin{figure}[t]
\begin{small}
\begin{math}
\begin{array}{c}
\begin{array}{cc}
\NamedRule{\varfut}
            {}
            {
                \veeval{\activationrecord, \assocval{\var}{\futid}}{\var}
                \vestep
                \veeval{\activationrecord}{\futid}
            }
            {}
\\[2ex] 
\NamedRule {\varother}          
            {}
            {
                \veeval{\activationrecord, \assocval{\var}{\val}}{\var}
                \vestep
                \veeval{\activationrecord, \assocval{\var}{\val}}{\val}
            }
            {\val \neq \gradedr{\resid}{\ggradeid}, \futid}
& 
\NamedRule{\veval}{}
{ \veeval{\activationrecord}{\val} \vestep \veeval{\activationrecord}{\val} } 
{} 
\\[3ex]
\hline
\\
\end{array}
\\
\NamedRule {\expop}
                {
                    \veeval{\activationrecord_i}{\valueexpression_i}
                    \vestep
                    \veeval{\activationrecord_{i+1}}{\val_i}\quad \forall i\in 1..n 
                }
                {
                    \expeval{\activationrecord_1}{ \opexp{\many{\valueexpression}} }
                    \expstep{\emptylabel}
                    \expeval{\activationrecord_{n+1}}{\return{\val}}
                }
                { \many\ve = \ve_1,\ldots,\ve_n \\ 
                  \defin{\key{op}}(\many\val)= \val  }
\\[4ex]
\begin{array}{cc}
  \NamedRule{\expchoiceleft}
                    { }
                    {
                        \expeval{\activationrecord}{ \nondetchoic{\e_1}{\e_2}}
                        \expstep{\emptylabel}
                        \expeval{\activationrecord}{ {\e_1}}
                    }{}
&
\NamedRule{\expchoiceright}
                    {}
                    {
                       \expeval{\activationrecord}{ \nondetchoic{\e_1}{\e_2}}
                        \expstep{\emptylabel}
                        \expeval{\activationrecord}{ {\e_2}}
                    }{}
\end{array}
\\[2ex]
\NamedRule{\expletred}{
                            \expeval{\activationrecord}{\expression_1}
                            \expstep[\labelid_i]{\labelid_o}
                            \expeval{\activationrecord'}{\expression_1'}
                        }
                        {
                            \expeval{\activationrecord}{\letexp{\var}{\expression_1}{\expression_2}}
                            \expstep[\labelid_i]{\labelid_o}
                            \expeval{\activationrecord'}{\letexp{\var}{\expression_1'}{\expression_2}}
                        }
              {}

\end{array}
\end{math}
\end{small}
\caption{Missing evaluation rules for  value expressions and expressions}    \label{fig:expeval:app}\label{fig:rulesve:app}
\end{figure}

\begin{figure}[!th]
\begin{small}
\begin{math}
\begin{array}{c}
\NamedRule{\procsilent}
               {
                \expeval{\activationrecord}{\expression}
                \expstep{\emptylabel}
                \expeval{\activationrecord'}{\expression'}
            }
            {
                \Conf{\actorenvironment}{\activethread{\activationrecord}{\expression}{\futid}{\actorid}}
                \semanticstep
                \Conf{\actorenvironment}{\activethread{\activationrecord'}{\expression'}{\futid}{\actorid}}
            }
            {}
 \\[3ex]
\NamedRule{\procreturn}
            {
                \veeval{\activationrecord}{\valueexpression}
                \vestep
                \veeval{\activationrecord'}{\val}
            }
            {
                \Conf{
                    \actorenvironment
                }{
                    \activethread{\activationrecord}{\return{\valueexpression}}{\futid}{\actorid}
                }
                \semanticstep
                \Conf{
                    \parallelcomp{
                    \actorenvironment
                }{
                    \placefut{\futid}{\val}
                }
                }{\idleactor{\actorid}}
                
            }
            {}
\\[3ex]
\hline
\\
\begin{array}{cc}
\NamedRule{\proclpar}
            {
                \Conf{\actorenvironment}{\processid}
                \semanticstep
                \Conf{\aactorenvironment}{\processid'}
            }
            {
                \Conf{\actorenvironment}{
                    \parallelcomp{\processid}{\pprocessid}
                }
                \semanticstep
                \Conf{\aactorenvironment}{
                    \parallelcomp{\processid'}{\pprocessid}
                }
            } {}
&
\NamedRule{\procrpar}
            {
                \parallelcomp{\actorenvironment}{\pprocessid}
                \semanticstep
                \parallelcomp{\aactorenvironment}{\pprocessid'}
            }
            {
                \parallelcomp{\actorenvironment}{
                    \parallelcomp{\processid}{\pprocessid}
                }
                \semanticstep
                \parallelcomp{\aactorenvironment}{
                    \parallelcomp{\processid}{\pprocessid'}
                }
            }
            {}       
\end{array}
\\[3ex]
\NamedRule{\proccong}
            {
                \Conf{\actorenvironment}{\PQ}
                \semanticstep
                \Conf{\aactorenvironment'}{\PP'}
            }
            {
                \Conf{\actorenvironment}{\PP} 
                \semanticstep
                \Conf{\aactorenvironment'}{\PP'}
            } {\PP\Pcong\PQ }

\end{array}
\end{math}
\end{small}
\caption{Missing evaluation rules for processes and configurations.}    \label{fig:processeval:app}
\end{figure}

For every process $\processid$, the sets 
$\fp\processid$ and $\fr\processid$ of futures \emph{produced} and \emph{consumed}, respectively, is
defined by:
\begin{align*} 
\fp{\PP} &= \begin{cases}
\{\futid\} & \text{if } \PP = \activethread{\_}{\_}{\futid}{\actorid}\text{ or }
                        \PP = \queuedthread{\_}{\_}{\futid}{\actorid}\text{ or } 
                        \PP= \placefut{\futid}{\_}  \\ 
\emptyset & \text{if }\PP = \pdone \\ 
\fp{\PP_1}\cup\fp{\PP_2} &  \text{if }\PP=\parallelcomp{\PP_1}{\PP_2}
\end{cases}
\\ 
\fr{\PP} &= \begin{cases}
\fr\activationrecord\cup\fr\e 
  & \text{if }\PP = \activethread\activationrecord{\e}{\futid}{\actorid}\text{ or } \PP = \queuedthread\activationrecord{\e}{\futid}{\actorid}\\
\bigcup_{i\in 1..n}\fr{\val_i}& \text{if } \PP= \callmsg{\futid}{\actorid}{\methid}{\val_1,\dots,\val_n} \\
\fr{\val}&  \text{if }\PP=\placefut{\futid}{\val}  \\
\emptyset &\text{if }\PP = \pdone \\ 
\fr{\PP_1}\cup(\fr{\PP_2}\setminus\fr{\PP_1}) & \text{if } \PP=\parallelcomp{\PP_1}{\PP_2}
\end{cases}
\end{align*}
where 
\begin{itemize}
\item $\fr{\assocval{\var_1}{\val_1},\ldots,\assocval{\var_n}{\val_n}} = \bigcup_{i\in 1..n} \fr{\val_i}$ and 
\item $\fr\e$ is the union of all $\fr\ve$ for all value expression $\ve$ occurring in $\e$\ and 
\item $\fr{\val} =\{\futid\}$ if $\val = \futid$, otherwise $\fr{\val} =\emptyset$.
\end{itemize}
In \cref{fig:expeval:app} are the evaluation rules for value expression (top part) and 
expressions missing from the main text. \\
Rule \refRule\varfut states that a variable containing a future
reduces to the latter and is removed from the local environment,
ensuring that futures are used \emph{linearly},
which is another distinctive feature of our calculus. 
Indeed, futures can typically  be read an arbitrary number of times;
however, since in our context they may contain graded resources, 
using futures in an unrestricted way 
could violate the constraint on the usage of the resource expressed by its grade.
Rule \refRule\varother reduces a variable to its value when it is neither a future nor a resource,  and 
Rule \refRule\veval reduces a value to itself and in both cases the resource context remains unchanged. 
Note that a variable containing a resource can reduce only if it is graded, i.e., if the needed amount of the resource is specified. 
This is needed to make resource consumption deterministic, and these annotations 
can be generated by the type checking process. \\
Rule \refRule\expop evaluates a primitive operation to its corresponding value according to its predefined semantics. The value returned cannot be a future. 
%
Rules \refRule\expchoiceleft and \refRule\expchoiceright reduce a non-deterministic choice to either one of the two branches. 
Finally, rule \refRule\expletred is the standard propagation rule for this construct. 

\noindent
In \cref{fig:processeval:app} are the evaluation rules for processes and configurations missing from the main text.\\
Rule \refRule\procsilent propagates silent steps of expressions to configurations, while 
Rule \refRule\procreturn terminates the execution of a thread transforming it into a fulfilled future with the resulting value. 
Rules \refRule{Lpar} and \refRule{Rpar} propagate reductions in a parallel composition and 
Rule \refRule\proccong closes the reduction under the precongruence on processes.

\subsection{\cref{sec:typesystem}}
\label{app:typesystem}
\begin{figure}[th]
 \begin{small}
\begin{math}
\begin{array}{c}
  \NamedRule{\tyval}
  {}
  {
    \WTvexpr\emptycontext{\typingcontext}{\unit}\Unit
  }  
  {\dis{\typingcontext} }
\quad
\NamedRule {\tyvalR}
  {}
  {
    \WTvexpr\emptycontext{\typingcontext}{\gradedr{\resid}{\ggradeid}}\res
    }  
    {\gradeid \gord \ggradeid\\ \dis{\typingcontext} }
\\[2ex]
\NamedRule{\tyvar}
  {} 
  {
     \WTvexpr{\emptycontext}{\var:\T,\typingcontext}{\var}\T
   } 
   { \T\neq\res \\
    \dis{\typingcontext}
   }
\quad
\NamedRule{\tyvalF}
  { }
  {
    \WTvexpr{\futid:\T}{\typingcontext}{\futid}\T
  }  
  {\dis{\typingcontext} }
\\[4ex]
\hline
\\[-2ex]
\NamedRule{\tyop}
  { 
    \WTvexpr{\futcontext_i}{\typingcontext_i }{\ve_i}{\T_i}\quad i\in 1..n
  }
  {
   \WTexpr{\actorcontext}{ \futcontext}{\typingcontext}{1}{\actorid}{\opexp{\many\ve}}{\T}{\actorcontext}
   }
    {  
     \futcontext = \futcontext_1 \linearsum \dots \linearsum \futcontext_n\\
      \typingcontext = \typingcontext_1 \tysum \cdots \tysum \typingcontext_n \\
       \OpTy{\key{op}}{\many\T}{\T} 
  }
  \\[4ex]
 \NamedRule{\tychoice}
  {
    \WTexpr{\actorcontext}{ \futcontext}{\typingcontext}{n_i}{\actorid}{\e_i}{\T}{\aactorcontext}\quad i=1,2 
   }
    {
        \WTexpr{\actorcontext}{ \futcontext}{\typingcontext}{1+p}{\actorid} {\nondetchoic{\e_1}{\e_2}}{\T}{\aactorcontext}
    }
   {p=n_1\ \vee\ p=n_2}         
 \\[4ex]   
   \NamedRule {\tyreturn}
   {
     \WTvexpr{\futcontext}{\typingcontext}{\ve}{\T}
   }
   {
        \WTexpr{\actorcontext}{ \futcontext}{\typingcontext}{0}{\actorid}{\return{\valueexpression}}{\T}{\actorcontext}
    }
    {}
   \\[4ex]
\hline
\\[-2ex]
 \NamedRule{\tyidle}
      {}
      {
       \WTproc\emptycontext{\emptycontext}{0}{\idleactor{\actorid} }{\emptycontext}
      }
       {}
\quad
 \NamedRule{\tydone}
      {}
      {
       \WTproc\emptycontext{\emptycontext}{0}{\pdone}{\emptycontext}
      }
       {}
\end{array}    
\end{math}
\end{small}
\caption{Missing typing rules for value expressions, expressions and processes}
\label{fig:types4expr:app}
\end{figure}
In \cref{fig:types4expr:app} are the typing rules for value expressions, expressions and processes missing from the main text. \\
Rules for values are straightforward, noting that in Rule \refRule\tyvalF the future is required to be the only one appearing in the future context with its type. 
Rule \refRule\tyvar extracts a variable of a non-resource type from the typing context. 
Rule \refRule\tyvarR types a graded variable $\gradedr\var\gr$ provided that $\var$ appears in the context with a grade $\ggr$ larger than $\gr$. 
Note that in all rules the future context is empty (except  for \refRule\tyvalF) and additional (not used) variables are allowed only provided that they are discardable. 
\\
Rule \refRule\tyop just checks that the arguments have the expected types and do not affect either of the actor contexts. 
Regarding this rule, we assume that every primitive operation $\key{op}$ in $\Sigop$ comes with a typing 
$\OpTy{\key{op}}{\T_1,\ldots,\T_n}{\T}$ 
where types cannot be futures and are such that, 
whenever $\WTvexpr\emptyctx\emptyctx{\val_i}{\T_i}$ for all $i \in 1..n$, then 
$\defin{\key{op}}({\val_1,\ldots,\val_n}) = \val$ 
and $\WTvexpr\emptyctx\emptyctx\val\T$. 
Again this guarantees that the typing and the dynamic semantics of primitive operations are consistent. 
Rule \refRule\tychoice checks that the two branches  of a non-deterministic choice are typed with the same type and contexts, which are also assigned to the whole expression. 
The measure is~$1$ plus the measure of one of the two branches. 
This allows not completely forbidding non-termination, as we only ensure that the measure decreases  in one of the two branches, while in the other can become arbitrarily large. 
Finally, Rule \refRule\tyreturn checks that the return value expression is well-typed and does not affect the actor contexts. 
The measure of the \key{return} is $0$ as it cannot reduce. 
\\
Rules \refRule\tyidle and \refRule\tydone state that an idle actor and $\pdone$ are well-typed in the empty contexts and with measure $0$.



\section{Proof of Fair Termination}
\label{apx:proof}

Firstly, we prove type preservation. 
This result ensures that any step taken in a well-typed configuration will reduce to a well-typed configuration. 
\par 
To denote correct steps we introduce typing of labels, that will only be used for the subject reduction proof.
The typing rules for expressions in local environments and labels used in the proofs are presented in \cref{fig:types4conf} and \cref{fig:types4labels} respectively.

\begin{figure}[th]
\begin{math}
\begin{array}{c}
 \NamedRule{\tyemptyenv}
   { }
   {
        \WTenv{\futcontext}{\emptycontext}{\emptycontext}
   }
   {}
   \quad   \quad
\NamedRule{\tyenv}
  { \WTenv{\futcontext}{\activationrecord}{\typingcontext}\quad \WTvexpr{\futcontext'}{\emptycontext}{\val}{\T}}
   {
        \WTenv{\futcontext,\futcontext'}{\assocval\var\val,\activationrecord}{\var:\T,\typingcontext}
   }
   {}
 \\[4ex]
 \NamedRule{\tyve}
   {
      \WTenv{\futcontext}{\activationrecord}{\typingcontext\gsum\ttypingcontext}
       \quad
      \WTvexpr{\futcontext'}{\typingcontext}{\ve}{\T}
   }
   {
      \WTvectx{\ttypingcontext}{\futcontext, \futcontext'}{\typingcontext}
      {\veeval{\activationrecord}{\ve}}{\T}
   }
   {}
   \quad
\NamedRule{\tyexp}
  { 
    \WTenv\futcontext\activationrecord{\typingcontext\gsum\ttypingcontext}\quad
    \WTexpr{\actorcontext}{ \futcontext'}{\typingcontext}{n}{\actorid}{\e}{\T}{\actorcontext'}
  }
  {
     \WTexpctx{\ttypingcontext}\actorcontext{\futcontext,\futcontext'}{\typingcontext}{n}\actorid{\expeval{\activationrecord}{\e}}\T\actorcontext'
  }
  {}
\end{array}
\end{math}
\caption{Rules for expressions in local environment}
\label{fig:types4conf}
\end{figure}

\begin{figure}[th]
    \begin{math}
        \begin{array}{c}
         \NamedRule{\tylblsilent}
           {}
          {
            \WTlab\emptycontext\emptycontext{\actorid}\emptylabel\emptycontext\emptycontext
          }
         {}
         \\[4ex] 
        \NamedRule{\tylblrls}
          {}
          {
          \WTlab\emptycontext\emptycontext{\actorid}{\rlsmsg{\gradedr{\resid}{\gradeid}}}\emptycontext {\actorid:{\res}}
          }
          {}
          \quad\quad
        \NamedRule{\tylblhold}
          {}
          {
          \WTlab{\actorid:{\res}}\emptycontext{\actorid}{\holdmsg{\gradedr{\resid}{\gradeid}}}\emptycontext \emptycontext
          }
          {}
         \\[4ex] 
        \NamedRule{\tylblcall}
           {
                       \WTvexpr{\futcontext_i}{\emptycontext}{\val}{\T_i}\quad i\in 1\dots n
            }
         {
           \WTlab{ \actorcontext}{\futcontext_1 \linearsum \dots \linearsum \futcontext_n}{\actorid}{ \callmsg{\futid}{\actorid}{\methid}{\many{\val}}}  {\futid:  \FutTy{\T}{\aactorcontext}} 
           \emptycontext
          }
          { \mtype{\actorid}{\methid} =\MType{\actorcontext}{\many\T}{\T}{n}{\aactorcontext}}
          \\[4ex] 
        \NamedRule{\tylblresult}
          {
            \WTvexpr{\futcontext}{\emptycontext}{\val}{\T}
          }
          {
             \WTlab\actorcontext{\futcontext} {\actorid}{\placefut{\futid}{\val}}{\futid:  \FutTy{\T}{\actorcontext}} \emptycontext
        }
        {}
         \end{array}
        \end{math}
        \caption{Typing rules for labels}
         \label{fig:types4labels}         
\end{figure}

\begin{lemma}[Subject reduction for value expressions] 
  \label{lem:srve}
    If $\WTvectx{\typingcontext'}{\futcontext}{\typingcontext}{\veeval{\activationrecord}{\valueexpression}}{\T}$,
    and $\veeval{\activationrecord}{\valueexpression} \vestep \veeval{\activationrecord'}{\valueexpression'}$. 
    Then: $\WTvectx{\typingcontext'}{\ffutcontext}{\ttypingcontext}{\veeval{\activationrecord'}{\valueexpression'}}{\T}$ 
\end{lemma}

\begin{proof}
    In order to apply (\tyve) $\ffutcontext$ is split into $\ffutcontext_1,\ffutcontext_2$ and by inversion of (\tyval) on hypothesis we know H1: $\WTenv{\futcontext_1}{\activationrecord,\assocval{\var}{\futid}}{\typingcontext_{\typingcontext'}}$.
    We proceed by induction on the reduction. 
    \begin{itemize}
        \item (\varres)
            $ \veeval{
                                \activationrecord, \assocval{\var}{\gradedr{\resid}{\ggradeid}}
                            }
                            {\gradedvar{\var}{\gradeid}}
                        \vestep
                            \veeval{
                                \activationrecord, \assocval{\var}{\gradedr{\resid}{\ggradeid \gsub \gradeid}}
                            }
                            {\gradedr{\resid}{\gradeid}}
                        $ 
            By well-typeness hypothesis (\tyve) holds for a certain $\futcontext=\futcontext_1,\futcontext_2$. 
            By (\tyvarR) $\futcontext_2=\emptycontext$ and $\typingcontext=\typingcontext'',\assoctype{\var}{\gradedr\resid\ggradeid}$, with $\dis{\typingcontext'}$.
            Therefore, we define $\ffutcontext_2=\emptycontext$ and $\ffutcontext_1=\futcontext_1$. 
            By (\tyvalR) with $\ffutcontext_2$ and $\ttypingcontext=\typingcontext'',\assoctype{\var}{\ggradeid\gsub\gradeid}$. 
            Where $\ttypingcontext$ is discardable due to all of its elements being discardable.
            Finally, $\WTenv{\ffutcontext_1}{\activationrecord,\assocval{\var}{\gradedr{\resid}{\ggradeid\gsub\gradeid}}}{\ttypingcontext\gsum\typingcontext'}$ by (\tyenv) using $\typingcontext''\gsum\ttypingcontext'$ for $\activationrecord$ and $\assoctype{\var}{\gradedr\resid\ggradeid}$ for $\var$. 
        \item (\varfut) $
        \veeval{\activationrecord, \assocval{\var}{\futid}}{\var}
        \vestep
        \veeval{\activationrecord}{\futid}
        $
            Applying (\tyvalF) on $\futid$ requires $\ffutcontext_2=\assoctype{\futid}{\T}$ and $\dis{\ttypingcontext}$.
            By inversion of (\tyenv) on H1: $\WTenv{\futcontext_1'}{\activationrecord}{\typingcontext''_{\typingcontext'}}$, where $\futcontext_1 = \futcontext_1',\assoctype{\futid}{\T}$ and $\typingcontext=\typingcontext'',\assoctype{\var}{\T}$.
            By (\tyvar) we know $\typingcontext''$ and therefore we can
            define $\ffutcontext_1=\futcontext_1'$ and 
            $\ttypingcontext=\typingcontext''$ to apply (\tyve) on:
            $\WTvectx{\typingcontext'}{\ffutcontext}{\ttypingcontext}{\veeval{\activationrecord}{\futid}}{\T}$.
        \item (\varother) $\veeval{\activationrecord, \assocval{\var}{\val}}{\var}
        \vestep
        \veeval{\activationrecord, \assocval{\var}{\val}}{\val}$
            By sidecondition of (\tyve) $\val=\unit$, therefore we can apply (\tyve) by the same premise as the inversion of (\tyve) on the hypothesis, plus (\tyval) for the value expression.
        \item (\veval) $\veeval{\activationrecord}{\val} \vestep \veeval{\activationrecord}{\val}$ trivial.
    \end{itemize}
\end{proof}

\begin{lemma}[Subject reduction of expressions] \label{lem:sre}	
    If (*)$\WTexpctx{\typingcontext'}{\actorcontext}{\futcontext}{\typingcontext}{n}{\actorid}{\expeval{\activationrecord}{\expression}}{\T}{\actorcontext'}$
    and 
    (**)$\expeval{\activationrecord}{\expression} \expstep[\labelid']{\labelid} \expeval{\activationrecord'}{\expression'}$ 
    and (***)$\WTlab{\actorcontext''}{\futcontext'}{\actorid}{\labelid'}{\futcontext''}{\aactorcontext''}$
    $\implies$
    \begin{itemize}
        \item $\WTexpctx{\typingcontext'}{\aactorcontext}{\ffutcontext,\ffutcontext_1'}{\ttypingcontext}{m}{\actorid}{\expeval{\activationrecord'}{\expression'}}{\T}{\actorcontext'}$
        \item $\WTlab{\aactorcontext_1}{\ffutcontext_1}{\actorid}{\labelid}{\ffutcontext_1'}{\aactorcontext_1'}$.
    \end{itemize}
    Where: 
    \begin{itemize}
        \item (a) $\aactorcontext \gsum \aactorcontext_1 \gsum \actorcontext'' \gord \actorcontext \gsum \aactorcontext_1' \gsum \aactorcontext''$ 
        \item (b) $\futcontext = \ffutcontext, \ffutcontext_1, \futcontext'$
    \end{itemize}
\end{lemma}

\begin{proof}
    By (\tyexp) $\futcontext = \futcontext_1,\futcontext_2$.
    By induction on the reduction rules of expressions.
    The base cases are the following:
    \begin{itemize}
        \item \textbf{(exp-await)}
            $\expeval{\activationrecord}{\awaitexp{\valueexpression}}
            \expstep[\placefut{\futid}{\val}]{}
            \expeval{\activationrecord'}{\return{\val}}$. 

            By hypothesis we know that (\tyawait) holds then: 
            $
            \actorcontext \joinctx \futcontext \joinctx \typingcontext 
            \vdash_{\actorid}^{n} 
            \assoctype{\awaitexp{\valueexpression}}
            {\T \joinctx \actorcontext \gsum \actorcontext''} 
            $, given that H1:
            $\futcontext \joinctx \typingcontext \vdash_{\actorid} 
            \assoctype{\valueexpression}{\FutTy{\T}{\actorcontext''}}$, 
            so 
            $\actorcontext' = \actorcontext \gsum \actorcontext''$,
            and by (\tyvarF) 
            $\futcontext = \assoctype{\var}{\FutTy{\T}{\actorcontext''}}$ and 
            $\typingcontext = \assoctype{\var}{\FutTy{\T}{\actorcontext''}}$.
            The first goal is to type $\expression'$: 
            $\actorcontext' \joinctx \ffutcontext \joinctx \ttypingcontext 
            \vdash_{\actorid}^m \assoctype{\return{\val}}
            {\T \joinctx \actorcontext'}$.  
            Due to (***) and (b): $\ffutcontext$ must be $\emptycontext$ 
            and due to $\labelid$ being silent we know that $\ffutcontext_1' = \emptycontext$.
            Then to apply the (\tyresult) rule we need to show that:
            $\emptycontext \joinctx \ttypingcontext \vdash \assoctype{\val}{\T}$.
            We do this using the premise of (\tylblresult) with empty context tells us that: $\vdash \assoctype{\val}{\T}$.
            We can apply (\tyval) showing that $\ttypingcontext$ is 
            discardable.
            This follows from Subject Reduction of value expressions applied to 
            $\veeval{\activationrecord}{\var} \vestep \veeval{\activationrecord'}{\futid}$ with residual environment $\typingcontext'$ and H1:
            $
            \WTvectx{\typingcontext'}{\assoctype{\futid}{\FutTy{\T}{\actorcontext''}}}{\ttypingcontext}{\veeval{\activationrecord'}{\futid}}{\T}$ and 
            $\emptycontext \vdash \assocctx{\activationrecord'}{\ttypingcontext'}$ where 
            $\activationrecord' = \activationrecord, \assocval{\var}{\futid}$. 
            By sidecondition of 
            $\assoctype{\futid}{\FutTy{\T}{\actorcontext''}} \joinctx\ttypingcontext' 
            \vdash \assoctype{\futid}{\T}$ we know that $\ttypingcontext'$ is discardable, and by premise 
            of (\tyve) we can define $\ttypingcontext= \ttypingcontext'$.
            Finally by (***) $\aactorcontext'' = \actorcontext''$ and substitution proves (a).
        \item \textbf{(exp-hold)}
        $\expeval{\activationrecord}{\holdexp{\gradeid}{\resid}}
         \expstep{\holdexp{\gradeid}{\resid}}
         \expeval{\activationrecord}{\return{\gradedr{\resid}{\gradeid}}}$.
        By (\tyhold) $\WTexpr{\assocresenv{\actorid}{\gradedr{\resid}{\gradeid}} \gsum \actorcontext'}{\emptycontext}{\emptycontext}{n}{\actorid}{\holdexp{\gradeid}{\resid}}{\gradedr{\resid}{\gradeid}}{\actorcontext'}$, and by (\tyenv) 
        $\WTenv{\futcontext}{\activationrecord}{\typingcontext \gsum \typingcontext'}$.
        Then we apply (\tyvalR) $\WTvexpr{\emptycontext}{\emptycontext}{\gradedr{\resid}{\gradeid}}{\gradedr{\resid}{\gradeid}}$ and (\tyreturn): $\WTexpr{\actorcontext'}{\futcontext}{\typingcontext \gsum\typingcontext'}{n}{\actorid}{\return{\gradedr{\resid}{\gradeid}}}{\gradedr{\resid}{\gradeid}}
        {\actorcontext'}$. 
        The typing of the label is: 
        $
        \WTlab{\actorid:{\res}}\emptycontext{\actorid}{\holdexp{\gradedr{\resid}{\gradeid}}}\emptycontext \emptycontext
        $ and $\labelid$ is $\emptylabel$ so (a) is proven by substitution and reflexivity. 
        Finally (b) holds by $\futcontext'=\emptycontext$ and $\ffutcontext_1'=\emptycontext$.
        
        \item \textbf{(exp-release)}
            $
            \expeval{\activationrecord}{\releaseexp{\gradeid}{\valueexpression}}
            \expstep{\rlsmsg{\gradedr{\resid}{\gradeid}}}
            \expeval{\activationrecord'}{\return{\unit}}
            $.
            Applying \cref{lem:srve} on the premise of (\exprls) with residual context $\typingcontext'$, we get H1:
            $\WTvectx{\typingcontext'}{\ffutcontext}{\ttypingcontext}{\veeval{\activationrecord'}{\gradedr{\resid}{\ggradeid}}}{\gradedr{\resid}{\gradeid}}$ with $\gradeid\gord\ggradeid$,
            and by inverting (\tyvalR) on $\WTvexpr{\emptycontext}{\ttypingcontext}{\gradedr{\resid}{\ggradeid}}{\gradedr{\resid}{\gradeid}}$ we know $\dis\ttypingcontext$.
            We use $\ffutcontext$ and $\ttypingcontext$ to type the new expression environment with residual $\typingcontext'$:
            $\WTexpctx{\typingcontext'}{\actorcontext'}{\ffutcontext}{\ttypingcontext}{0}{\actorid}{\return{\unit}}{\Unit}{\actorcontext'}$, by using the environment premise of H1 plus (\tyresult) with $\dis\ttypingcontext$.
            The label is typed by (\tylblrls) and (a) holds by substitution: $\actorcontext' \gord \actorcontext\gsum\assocresenv{\actorid}{\gradedr{\resid}{\gradeid}}$, where $=$ holds by (\tyrls) on hypothesis.
            Finally (b) hold by $\futcontext=\ffutcontext$.
        \item \textbf{\expcall}
            $
            \expeval{\activationrecord_1}{\callexp{\actorid}{\methid}{\many{\valueexpression}}}
            \expstep{\callmsg{\futid}{\actorid}{\methid}{\many\val}}
            \expeval{\activationrecord_{k+1}}{\return{\futid}}
            $.
            Applying \cref{lem:srve} iteratively on the arguments with residual typing context $\typingcontext'$ we obtain: $
            \WTvectx{\typingcontext'}{\ffutcontext_k}{\ttypingcontext_{k}}{\veeval{\activationrecord_{k}'}{\val_{k}}}{\T_k}$.
            According to the method table: $\MType{\aaactorcontext}{\many{T'}}{\T}{p}{\aaactorcontext'}$.
            Then we type using $\ffutcontext=\ffutcontext_k$ and $\ttypingcontext=\ttypingcontext_{k}$: 
            $\WTexpctx{\typingcontext'}{\actorcontext'}{\ffutcontext, \assoctype{\futid}{\FutTy{\T}{\aaactorcontext'}}}{\ttypingcontext}{0}{\actorid}{\expeval{\actorcontext_k'}{\return{\futid}}}{\FutTy{\T}{\aaactorcontext'}}{\actorcontext'}$
            by using (\tyreturn) with $\assoctype{\futid}{\FutTy{\T}{\aaactorcontext'}}$ and $\ttypingcontext_k$ (by \cref{lem:srve} we know it is discardable), 
            and (\tyenv) with $\ffutcontext$.
            The label is typed by (\tylblcall) making the goal (a):
            $\actorcontext' \gsum\aaactorcontext'\gord\actorcontext$,
            which holds with $(=)$ by (\tycall); 
            and (b):
            $\futcontext=\ffutcontext,\ffutcontext_1$, which holds by inversion of (\tyexp) and (\tycall) on well-typedness hypothesis.
        \item \textbf{\expop} $
            \expeval{\activationrecord_1}{ \opexp{\many{\valueexpression}} }
            \expstep{\emptylabel}
            \expeval{\activationrecord_{n+1}}{\return{\val}}
            $.
            Similarly to (\tycall), \cref{lem:srve} is applied iteratively on $\veeval{\activationrecord_i}{\valueexpression_i}$ with residual typing context $\typingcontext'$, obtaining H1: $\WTvexpr{\ffutcontext_k}{\ttypingcontext_{k,\typingcontext'}}{\veeval{\activationrecord_{k+1}}{\return{\val}}}{\T'}$.
            The label is typed with (\tylblsilent).
            To apply (\tyexp) we define $\ffutcontext=\ffutcontext_k$ and $\ttypingcontext_{\typingcontext'}=\ttypingcontext_{k,\typingcontext'}$, and by inversion of (\tyenv) on H1 plus (\tyresult):
            $\WTexpr{\actorcontext'}{\ffutcontext}{\ttypingcontext}{1}{\actorid}{\return{\val}}{\T}{\actorcontext'}$. 
            Where well-typedness of $\val$ is given either by (\tyvalR) or (\tyval).
            Finally, (a): $\actorcontext'\gord\actorcontext$ is given by (\tyop), and (b) $\futcontext=\ffutcontext=\emptycontext$. 
            This is the case due to well-formedness assumption of $\opexp{}$. 
    \end{itemize}
    The inductive cases:
    \begin{itemize}
        \item \textbf{(exp-let)}
            $\expeval{\activationrecord}{\letexp{\var}{\expression_1}{\expression_2}}
            \expstep[\labelid']{\labelid}
            \expeval{\activationrecord'}{\letexp{\var}{\expression_1'}{\expression_2}}$.

            By hypothesis there are context splittings: 
            $\actorcontext = \actorcontext_1 \gsum \actorcontext_2$, 
            $\futcontext = \futcontext_1, \futcontext_2$,
            $\typingcontext = \typingcontext_1 \gsum \typingcontext_2$,
            $\actorcontext' = \actorcontext_1'' \gsum \actorcontext_2'$. 
            Applying I.H. on well-typeness of $\expression_1$ with residual typing context: $\typingcontext' \gsum \typingcontext_2$: 
            $IH_1$:
                $ \aactorcontext_I \joinctx \ffutcontext_I, \futcontext_{1I}'
                \joinctx \ttypingcontext_I \gsum \typingcontext' \gsum \typingcontext_2
                    \vdash_{\actorid}^m \assoctype{\expeval{\activationrecord'}{\expression_1'}}{\T' \joinctx \actorcontext_1'}
                $. 
            With a splitting $\ffutcontext_I, \futcontext_{1I}' = \ffutcontext_I', \ffutcontext_I''$, such that: $\ffutcontext_I' \vdash \assocctx{\activationrecord}{\ttypingcontext_I \gsum \typingcontext' \gsum \typingcontext_2}$,
            $IH_2$:
                $\aactorcontext_{1H} \joinctx \futcontext_{1H} \vdash_{\actorid} \assocctx{\labelid}{\futcontext_{1I}' \joinctx \aactorcontext_{1I}'}$. 
            Defining $\aactorcontext_{1H} = \aactorcontext_{1}, \ffutcontext_{1H} =\ffutcontext_1, \ffutcontext_{1I}' = \ffutcontext_{1}'$ and $\ttypingcontext_{1I}' = \ttypingcontext_{1}'$ we prove typing of the label.
            $IH_a$:
                $\aactorcontext_I \gsum \aactorcontext_1 \gord \actorcontext_1 \gsum \aactorcontext_1' \gsum \aactorcontext''$.
            $IH_b$:
                $\futcontext_1 = \ffutcontext_I, \ffutcontext_1', \futcontext'$.
            Using the premise of (*) and $IH_1$ and $IH_3$: 
                $\aactorcontext_I \gsum \actorcontext_2 \joinctx \ffutcontext_I, \ffutcontext_1' \joinctx \ttypingcontext_I \gsum \typingcontext'
                \vdash_{\actorid}^{m+n} \assoctype{\letexp{\var}{\expression_1'}{\expression_2}}{\T \joinctx \actorcontext'}$.
            By this result and $\ffutcontext'', \futcontext_2 \vdash \assocctx{\activationrecord'}{\typingcontext_I \gsum \typingcontext_2}$, we apply (\tyexp).
            (a): 
            $\aactorcontext_I \gsum \actorcontext_2 \gsum \aactorcontext_1 
            \gord 
            \actorcontext_1 \gsum \actorcontext_2 \gsum \aactorcontext_1' \gsum \aactorcontext''$, by $IH_a$ we have: $\actorcontext_2 \gord \actorcontext_2$.
            (b): $\futcontext = \futcontext_1, \futcontext_2 = \ffutcontext_I, \ffutcontext_1, \futcontext', \futcontext_2 = \ffutcontext, \ffutcontext_1, \futcontext'$ by $IH_b$.

        \item \textbf{(exp-let2)}
            $\expeval{\activationrecord}{\letexp{\var}{\return{\valueexpression}}{\expression_2}}
            \expstep{\emptylabel}
            \expeval{\activationrecord', \assocval{\vvar}{\val}}{\varbind{\expression_2}{\var}{\vvar}}$.
            By (T-let) we know that $\valueexpression$ types to $\T$ and that: 
            $
            \actorcontext_2 \gsum \actorcontext_1'' 
                    \joinctx 
                    \futcontext_2 
                    \joinctx 
                    \typingcontext_2, \assoctype{\var}{\T}
                    \vdash_{\actorid}^{n}
                    \assoctype{
                        \expression_2
                        }
                        {
                            \T' \joinctx \actorcontext_2'
                        }
            $.
            Then SR holds by I.H. on $\varbind{\expression_2}{\var}{\vvar}$.
        \item \textbf{(exp-ndet) and (exp-ndet2)}
            By I.H. on the subterm.
        
    \end{itemize}
\end{proof}

The precongruence of configurations ($\Pcong$) is used to make bureaucratic progress in a configuration,
these are computationally irrelevant steps such as moving a $\idleactor{\actorid}$ process throughout the configuration, swapping independent processes, idling awaiting threads and restarting idle threads.
These last two characterize the cooperative schedule paradigm. 
It is important to prove that this relation does not affect type nor measure of configurations, due to its computational irrelevance.

\begin{definition}[Precongruence of configurations]
  Given $\config = \Conf{\actorenvironment}{\processid}$ and $\config' = \Conf{\actorenvironment}{\processid'}$, if $\processid \Pcong \processid'$ then $\config \Pcong \config'$.
\end{definition}

\begin{lemma}[Subject precongruence]\label{lem:spre}
    Given configurations $\config$ and $\config'$ such as $\config \Pcong \config'$, 
    if $\actorcontext \joinctx \futcontext \vdash^{p} \assocctx{\config}{\futcontext}$ then 
    $\actorcontext \joinctx \futcontext \vdash^{p} \assocctx{\config'}{\futcontext}$
\end{lemma}

\begin{proof}
    By induction on the congruence rule.
    \begin{itemize}
        \item (\cswap) $\processid \parP \pprocessid \Cong \pprocessid \parP \processid$. 
        ($\implies$) By side condition $\futcontext_1'=\emptycontext$, $\futcontext=\futcontext_p,\futcontext_q$.
        So it is trivial to see that (\typarallel) is applicable for $\parallelcomp{\pprocessid}{\processid}$,
        with $\futcontext_1=\emptycontext$, the side condition of (\typarallel) holds by side condition of (\cswap),
        and the measure $p=n+m=m+n$.
        The ($\impliedby$) is analogous.
        \item (\cactivate) $\parallelcomp{\idleactor{\actorid}}{
            \queuedthread{\activationrecord}{\expression}{\futid}{\actorid}
        }
        \Pcong
        \activethread{\activationrecord}{\expression}{\futid}{\actorid}$
        By (\tyidle) $\futcontext_1',\futcontext_1''=\emptycontext$, $\futcontext_1=\emptycontext$ and $\actorcontext=\emptycontext$ and $n=0$.
        So by (\typarallel) $\futcontext=\futcontext_2$, $\actorcontext=\actorcontext_1$, $\futcontext'=\futcontext_2'$ and $m=p$.
        So (\tythread) applies on $\WTproc{\actorcontext}{\futcontext}{p}{\activethread{\actorcontext}{\e}{\futid}{\actorid}}{\futcontext'}$.
        \item (\cyield) 
        $\activethread{\activationrecord}{\expression}{\futid}{\actorid}
        \Pcong
        \parallelcomp{\idleactor{\actorid}}{
            \queuedthread{\activationrecord}{\expression}{\futid}{\actorid}
        }$
        Application of (\typarallel) with $\actorcontext_1=\emptycontext$, $\futcontext_1=\emptycontext$ and 
        $\futcontext_1',\futcontext_1''=\emptycontext$ with measure $n=0$ holds with hypothesis,
        and the measure is $m=p$.
    \end{itemize}
\end{proof}

Now we have the tools to prove subject reduction of configurations as it was stated in \cref{thm:subred}. 

\begin{proof}
	We proof subject reduction as standard by induction over the reduction rules.
    We know that $\config$ is of the shape $\parallelcomp{\actorenvironment}{\processid}$
    and by (T-conf) $\processid$ has to be typed under a process typing rule.
	\begin{itemize}
		\item \textbf{\procget}:
            $\Conf{\actorenvironment}{
            \parallelcomp{
                \placefut{\futid}{\val}}{
                    \activethread{\activationrecord}{\expression}{\futid'}{\actorid}}}
            \semanticstep 
            \Conf{
                \actorenvironment
            }{
                \activethread{\activationrecord'}{\expression'}{\futid'}{\actorid}
            }$.
            \\ 
            By inversion of (\typarallel) and then 
            inversion on (\tyresult) 
            H1: $\WTproc{\actorcontext_1}{\futcontext_1}{0}{\placefut{\futid}{\val}}{\assoctype{\futid}{\FutTy{\T}{\actorcontext_1}}}$ and (\tythread)
            H2: $\WTproc{\actorcontext_2}{\futcontext_2}{n'}{\activethread{\activationrecord}{\expression}{\futid'}{\actorid}}{\actorcontext_2'}$. 
            Where $\assoctype{\futid}{\FutTy{\T}{\actorcontext_1}} = \ffutcontext_1, \ffutcontext_2$, such that: 
            $\futcontext_2 = \futcontext_2', \ffutcontext_1$ for some $\futcontext_2'$.
            By (\procget), $\expeval{\activationrecord}{\expression} \expstep[\placefut{\futid}{\val}]{} \expeval{\activationrecord'}{\expression'}$ 
            and $\veeval{\activationrecord}{\valueexpression} \vestep \veeval{\activationrecord'}{\futid}$, 
            implying: $\activationrecord = \activationrecord', \assocval{\var}{\futid}$, 
            hence inverting (\tyenv) on the inversion of (\tythread) on the premise of H2, 
            we know that $\futid$ must appear on $\activationrecord$, 
            so $\ffutcontext_2 = \emptycontext$.
            Given that result, we can apply SRE on the premise of (\tythread) of H2, using $\ffutcontext_1$ in $\futcontext_2$ to type: $\WTlab{\emptycontext}{\ffutcontext_1}{\actorid}{\placefut{\futid}{\val}}{\emptycontext}{\emptycontext}$. 
            Finally $\config'$ is well-typed applying (\tythread) on the results of SRE and (\tyconf) on the composition with the environment, and $\ffutcontext''=\emptycontext$.
        \item \textbf{\procrestart}
            $\Conf{\actorenvironment}{
                \parallelcomp{
                    \placefut{\futid}{\val}
                }{
                    \parallelcomp{
                        \queuedthread{\activationrecord}{\expression}{\futid'}{\actorid}    
                    }{
                        \idleactor{\actorid}
                    }
                }
            }
            \semanticstep
            \Conf{\actorenvironment}{
                \parallelcomp{
                    \placefut{\futid}{\val}
                }{
                    \activethread{\activationrecord}{\expression}{\futid'}{\actorid}    
                }
            }$.
            $\config'$ is well-typed by the premises of (\tythread) on $\config$.
		\item \textbf{\procsuspend}
            $\Conf{
                \actorenvironment
            }{
                \activethread{\activationrecord}{\expression}{\futid'}{\actorid}
            }
            \semanticstep 
            \Conf{\actorenvironment}{
            \parallelcomp{
                \idleactor{\actorid}
                }
                {
                    \queuedthread{\activationrecord}{\expression}{\futid'}{\actorid}
                }
            }$.
            $\config'$ is well-typed applying (\typarallel) with $\pprocessid=\idleactor{\actorid}$ and $\processid$ using the the same premises as $\config$. 
		\item \textbf{\prochold}
            $\Conf
            {
                \actorenvironment, \actorid: \resenv, \gradedr{\resid}{\ggradeid}
            }
            {
                \activethread{\activationrecord}{\expression}{\futid}{\actorid}
            }
            \semanticstep
            \Conf
            {
                \actorenvironment, \actorid: \resenv, \gradedr{\resid}{\ggradeid \gsub \gradeid}
            }
            {
                \activethread{\activationrecord'}{\expression'}{\futid}{\actorid}
            }$.
            By SRE on the premise of (\tythread) $\WTproc{\aactorcontext}{\ffutcontext}{m}{\activethread{\activationrecord'}{\expression'}{\futid}{\actorid}}{\aactorcontext'}$, by (\prochold) the label is $\holdexp{\gradeid}{\resid}$ so by (a) $\aactorcontext\gsum \assocresenv{\actorid}{\gradedr{\resid}{\gradeid}}\gord\actorcontext\gsum\emptycontext\gsum\emptycontext$, given $\ggradeid'$ the grade of resource $\resid$ on actor $\actorid$ on the context $\aactorcontext$ and $\ggradeid''$ the one on $\actorcontext$, (a) implies $\ggradeid' \gsum \gradeid \gord \ggradeid''$. 
            By (\tyconf) on $\config$: $\actorcontext\gsum\rescontext\gord\actorenvironment, \actorid: \resenv, \gradedr{\resid}{\ggradeid}$, so $\ggradeid'' \gord \ggradeid$, and by definition of graded substractive monoid $\ggradeid' \gord \ggradeid \gsub \gradeid$. Then $\aactorcontext \gsum \rescontext \gord \actorenvironment, \assocresenv{\actorid}{\resenv, \gradedr{\resid}{\ggradeid\gsub\gradeid}}$, and we can apply (\tyconf).
		\item \textbf{\procrls}
            $\Conf
            {
                \actorenvironment, a: \resenv, \gradedr{\resid}{\ggradeid}
            }
            {
                \activethread{\activationrecord}{\expression}{\futid}{\actorid}
            }
            \semanticstep
            \Conf
            {
                \actorenvironment, a: \resenv, \gradedr{\resid}{\ggradeid + \gradeid}
            }
            {
                \activethread{\activationrecord'}{\expression'}{\futid}{\actorid}
            }$. We apply (\tythread) by \cref{lem:sre}, and the side condition of (\tyconf) holds by 
            (a) of \cref{lem:sre}.
		\item \textbf{\procreturn}
            $\Conf{
                \actorenvironment
            }{
                \activethread{\activationrecord}{\return{\valueexpression}}{\futid}{\actorid}
            }
            \semanticstep
            \Conf{
                \parallelcomp{
                \actorenvironment
            }{
                \placefut{\futid}{\val}
            }
            }{\idleactor{\actorid}}$. Holds by trivially by (\tyconf) (\typarallel) and (\tyresult) (\tyidle) respectively.
		\item \textbf{\procspawn}
            $\Conf{
                {\actorenvironment}
            }{
              \parallelcomp{\callmsg{\futid}{\actorid}{\methid}{\many{\val}}}{\idleactor{\actorid}}
            }
            \semanticstep
           \Conf{\actorenvironment}{\activethread{\many{\var\mapsto\val}}{\expression}{\futid}{\actorid}}$.
            We apply (\tyconf) with the same contexts and use the premise of (\typarallel) for (\tymsg) to apply (\tythread) which applies by well formedness of the method table.
		\item \textbf{\procsilent} $\Conf{\actorenvironment}{\activethread{\activationrecord}{\expression}{\futid}{\actorid}}
        \semanticstep
        \Conf{\actorenvironment}{\activethread{\activationrecord'}{\expression'}{\futid}{\actorid}}$. We apply (\tythread) by \cref{lem:sre}, and the side condition of (\tyconf) holds by 
        (a) of \cref{lem:sre}.
		\item \textbf{\proccall}
            $\Conf{\actorenvironment}
            {
                \activethread{\activationrecord}{\expression}{\futid'}{\aactorid}
            }
            \semanticstep
            \Conf{\actorenvironment}
            {
                \parallelcomp{\callmsg{\futid}{\actorid}{\methid}{\many{\val}}}{
                    \activethread{\activationrecord'}{\expression'}{\futid'}{\aactorid}
                }
            }$.
        By (\tythread) and \cref{lem:sre} $\WTthread{\aactorcontext}{\ffutcontext,\ffutcontext_1'}{\ttypingcontext}{\aactorid}{m}{\expeval{\activationrecord'}{\e'}}{\ffutcontext''}{\aactorcontext'}$,
        and $\WTlab{\aactorcontext_1}{\ffutcontext_1}{\aactorid}{\callmsg{\futid}{\actorid}{\methid}{\many{\val}}}{\ffutcontext_1'}{\emptycontext}$,
        where $\aactorcontext \gsum \aactorcontext_1 \gord \actorcontext$,
        $\mtype{\actorid}{\methid}=\MType{\aactorcontext_1}{\many{\T'}}{\T}{n'}{\aactorcontext''}$,
        and $\futcontext= \ffutcontext, \ffutcontext_1$.
        We can then apply (\typarallel) instantiating 
        $\actorcontext_2 = \aactorcontext$, 
        $\actorcontext_1 = \aactorcontext_1$,
        $\futcontext_1' = \ffutcontext_1'$,
        $\futcontext_1''=\emptycontext$,
        $\futcontext_2=\ffutcontext$ and 
        $\futcontext_1=\ffutcontext_1$. 
        Then for the full judgment $\ffutcontext'=\markcontext(\ffutcontext_1'),\ffutcontext_1'$ where $\ffutcontext_1'$ is fresh by side condition of (\expcall).
        The side condition of (\tyconf) holds by (a) of \cref{lem:sre}.
    \end{itemize}
    The inductive cases:
    \begin{itemize}
		\item \textbf{\proclpar}
        $\parallelcomp{\actorenvironment}{
            \parallelcomp{\processid}{\pprocessid}
        }
        \semanticstep
        \parallelcomp{\actorenvironment'}{
            \parallelcomp{\processid'}{\pprocessid}
        }$
        By (\typarallel): (*) $
        \WTproc{\actorcontext_1}{\futcontext}{n}{\processid}{\futcontext_1' \linearsum \futcontext_1'' }$ and (**) 
        $\WTproc{\actorcontext_2}{\futcontext_2\linearsum \futcontext_1''}{m}{\pprocessid}{\futcontext_2' }$. 
        By (\tyconf) on (*), using $\actorcontext_2\gsum\rescontext$ as residual context: 
        $
        \WTconf[\actorcontext_2\gsum\rescontext]{\actorcontext_1 }{\futcontext}{n}{\parallelcomp{\actorenvironment}{\processid}}{\futcontext_1' \linearsum \futcontext_1''}
        $,
         $\actorcontext_1 \gsum \actorcontext_2\gsum\rescontext \gord \actorenvironment$.
        by I.H.: 
        $\WTconf[\actorcontext_2\gsum\rescontext]{\actorcontext_1'}{\futcontext}{n'}{\parallelcomp{\actorenvironment'}{\processid'}}{\futcontext'',\markcontext(\ffutcontext'')}$,
         $\actorcontext_1' \gsum \actorcontext_2\gsum\rescontext \gord \actorenvironment'$.
        We apply (\typarallel) on 
        $
        \WTconf[\rescontext]{\actorcontext_1'+\actorcontext_2}
        {\futcontext \linearsum \futcontext_2 }
        {n'+m}{\config'}{\futcontext'',\markcontext(\ffutcontext''), \futcontext_2'}$ 
        with premises (**) and I.H., instantiating $\futcontext_1''=\markcontext(\ffutcontext'')$ such that the side condition of (\typarallel) holds by freshness of $\ffutcontext''$.
        And (\tyconf) with the satisfied side condition $\actorcontext_1'+\actorcontext_2\gsum\rescontext  \gord \actorenvironment'$.
		\item \textbf{\procrpar}
        Analogous to \proclpar.
        \item \textbf{\proccong}. Holds trivially by \cref{lem:spre} and I.H.
	\end{itemize}
\end{proof}

The typing measurement establishes a bound to the possible steps to a terminated configuration. 
We prove that for every configuration there is a possible reduction that reduces the measurement.

\begin{lemma}[Progress of expressions]\label{lem:progexp}
    If $\WTthread\actorcontext\futcontext\typingcontext{n}\actorid{\expeval{\activationrecord}{\e}}\T\actorcontext'$ and 0<n,
    then there exists a step
    $\expeval{\activationrecord}{\e} \expstep[\labelid']{\labelid} \expeval{\activationrecord'}{\e'}$
    such as $\WTthread\aactorcontext\ffutcontext\ttypingcontext{m}\actorid{\expeval{\activationrecord'}{\e'}}\T\aactorcontext'$, and 
    m<n.
\end{lemma}

\begin{proof}
    By induction on $\expression$. 
    For $\expression=\callexp{\actorid}{\methid}{\many{\ve}}$ holds by applying (\expcall), by subject reduction $\return{\futid}$ is well-typed with measure 0. For $\expression=\awaitexp{\ve}$ by (\expawait) and subject reduction is well-typed with measure 0. For $\expression=\holdexp{\gradeid}{\resid}$ reduces with (\exphold) and by subject reduction is well-typed with measure 0.
    For $\expression=\releaseexp{\gradeid}{\ve}$ reduces by (\tyrls) and by subject reduction is well-typed with measure 0. 
    For $\expression=\opexp{\many\ve}$ it reduces with (\tyop) and by subject reduction is well-typed with measure 0.
    The non-deterministic step $\expression=\nondetchoic{\e_1}{\e_2}$ is the only case with multiple reductions, by (\tychoice) $n=1+p$ where $p$ is the measure that types some $\e_i$, we perform that step, then by subject reduction $\e'$ is well-typed with measurement $p$.
    The first inductive step is $\expression=\letexp{\var}{\e_1}{\e_2}$ that by (\tylet) has measure $n=1+t+s$, if $\e_1 = \return{\ve}$ it reduces by (\explet) and is well-typed by subject reduction with measure $s$, if $\e_1 \neq \return{\ve}$ then it reduces by (\expletred) and is well-typed by subject reduction with measure $n = 1+t'+s$ where $t'<t$ by I.H. 
\end{proof}

The strategy to prove deadlock freedom is by showing that any possible stuck configuration can only be stuck waiting for a future. 
Then the configuration is either precongruent to a configuration that can take the (\procget) step, or to one that can make progress to fulfil that future.
Therefore, full configuration, typed on an empty future context, cannot be stuck.
The following lemmas show that a stuck thread has a possible reduction with a $\placefut{\futid}{\val}$ label, and that in a stuck configuration there is a precongruence that enables progress.

\begin{lemma}[Stuck thread]\label{lem:stuckthread}
    Given $\WTthread\actorcontext\futcontext\typingcontext{n}\actorid{\expeval{\activationrecord}{\e}}\T\actorcontext'$ and $\Conf{\actorcontext}{\activethread{\activationrecord}{\e}{\futid'}{\actorid}}\not\red$, then $\expeval{\activationrecord}{\expression} \expstep[\placefut{\futid}{\val}]{} \expeval{\activationrecord'}{\expression'}$ for some $\futid \in \dom(\futcontext)$
\end{lemma}

\begin{proof}
    By \cref{lem:progexp} there exists $\expeval{\activationrecord}{\e} \expstep[\labelid']{\labelid} \expeval{\activationrecord'}{\e'}$, by cases on $\expstep[\labelid']{\labelid}$, if $\expstep[]{\rlsmsg{\gradedr{\resid}{\gradeid}}}$: applying (\procrls) contradicts the hypothesis, if $\expstep[]{\callmsg{\futid'}{\actorid}{\methid}{\many\val}}$: applying (\proccall) contradicts the hypothesis. if $\expstep[\holdmsg{\gradedr{\resid}{\gradeid}}]{}$: applying (\prochold) is possible due to well-typeness hypothesis (\tyexp) and (\tyhold) \UTComm{(CAN I DO THIS? I NEED TO SHOW THAT IT ALSO HOLDS ON NESTED HOLDS)}, which contradicts the hypothesis, if $\expstep[]{}$ then (\procsilent) is possible which contradicts the hypothesis.
    Finally, if it is $\expstep[\placefut{\futid}{\val}]{}$ there is no $\semanticstep$ possible, and by inversion of (\tyexp) on the well-typeness hypothesis and (\tyenv) and (\tyvar) $\futid \in \dom(\futcontext)$.
\end{proof}

In order to manipulate parallel compositions of processes and be consistent with the associativity assumption, the typing of a configuration must be unaltered by associativity. 

\begin{lemma}[Associativity of the parallel composition]
    For all $\processid, \pprocessid, \ppprocessid$:
    $\WTproc{\actorcontext}{\futcontext}{n}{\parallelcomp{\processid}{(\parallelcomp{\pprocessid}{\ppprocessid})}}{\futcontext'} \implies \WTproc{\actorcontext}{\futcontext}{n}{\parallelcomp{(\parallelcomp{\processid}{\pprocessid})}{\ppprocessid}}{\futcontext'}$.
\end{lemma}
\begin{proof}
    $n=h+j+k$.
    We define $\futcontext=\futcontext_p,\futcontext_q,\futcontext_u$ and $\actorcontext=\actorcontext_p\gsum\actorcontext_q\gsum\actorcontext_u$.
    Where \\
    $\WTproc{\actorcontext_p}{\futcontext_p}{h}{\processid}
    {\futcontext_{\parallelcomp\processid\pprocessid},\futcontext_{\parallelcomp\processid\ppprocessid}, \futcontext_p''}$. \\
    $\WTproc{\actorcontext_q}{\futcontext_q,\futcontext_{\parallelcomp{\processid}{\pprocessid}}}{j}
    {\pprocessid}{\futcontext_{\parallelcomp{\pprocessid}{\ppprocessid}},\futcontext_q''}$.\\
    $\WTproc{\actorcontext_u}{\futcontext_u,\futcontext_{\parallelcomp{\processid}{\ppprocessid}},\futcontext_{\parallelcomp{\pprocessid}{\ppprocessid}}}
    {k}{\ppprocessid}{\futcontext_u'}$.\\
    By inversion of (\typarallel) on hypothesis: 
    $\futcontext_1' = \futcontext_{\parallelcomp\processid\pprocessid},\futcontext_{\parallelcomp\processid\ppprocessid}$ and 
    $\futcontext_1'' = \futcontext_p''$. 
    The side condition asserts H1: $\dom{(\futcontext_p'') \cap \dom{(\futcontext_q,\futcontext_u)}} = \emptycontext$;
    and the side condition on (\typarallel) of $(\parallelcomp{\pprocessid}{\ppprocessid})$ asserts H2: 
    $\dom{(\futcontext''_q)\cap\dom{(\futcontext_u)}}=\emptycontext$.
    We apply (\typarallel) on $\parallelcomp{\processid}{\pprocessid}$ the goal with: $\futcontext_1'=\futcontext_{\parallelcomp\processid\pprocessid}$ and $\futcontext_1''=\futcontext_p'',\futcontext_{\parallelcomp\processid\ppprocessid}$. 
    The side condition $\dom\futcontext_1''\cap\dom\futcontext_q=\emptycontext$ holds by linearity of $\futcontext_2=\futcontext_q,\futcontext_u,\futcontext_{\parallelcomp\processid\pprocessid}$ by inversion of (\typarallel) on the hypothesis.
    Therefore, $\WTproc{\actorcontext_p\gsum\actorcontext_q}{\futcontext_p,\futcontext_q}{h+j}{\parallelcomp{\processid}{\pprocessid}}{\futcontext'}$
    where
    $\futcontext'=\markcontext(\futcontext_{\parallelcomp\processid\pprocessid}),
    \futcontext''_p,\futcontext_{\parallelcomp\processid\ppprocessid},
    \futcontext_{\parallelcomp\pprocessid\ppprocessid},
    \futcontext_q''
    $.
    We apply (\typarallel) on $\parallelcomp{(\parallelcomp{\processid}{\pprocessid})}{\ppprocessid}$ with $\futcontext_1'=\futcontext_{\parallelcomp\processid\ppprocessid},\futcontext_{\parallelcomp{\pprocessid}{\ppprocessid}}$ and $\futcontext_1''=\futcontext_p'',\futcontext_q'',\markcontext(\futcontext_{\parallelcomp\processid\pprocessid})$.
    The side condition $\dom{(\futcontext_p'',\futcontext_q'',\markcontext(\futcontext_{\parallelcomp\processid\pprocessid}))} \cap \dom{(\futcontext_u)} = \emptycontext$ holds by $\futcontext''_p \cap \futcontext_u = \emptycontext$ from H1, $\futcontext''_q\cap\futcontext_u=\emptycontext$ from H2, 
    and $\dom{(\markcontext(\futcontext_{\parallelcomp\processid\pprocessid}))}\cap\futcontext_u=\emptycontext$ holds by linearity on the left-hand future context of ($\parallelcomp{\pprocessid}{\ppprocessid}$).
    The measure is $h+j+k$ which is $n$.
    Finally, we need to show that $\typingcontext=\futcontext_p'',\futcontext_q'',\markcontext(\futcontext_{\parallelcomp\processid\pprocessid}),\markcontext(\futcontext_{\parallelcomp\processid\ppprocessid},\futcontext_{\parallelcomp{\pprocessid}{\ppprocessid}}),\futcontext'_u$.
    By inversion on (\typarallel) of hypothesis $\futcontext=\markcontext(\futcontext_{\parallelcomp{\processid}{\pprocessid}},\futcontext_{\parallelcomp{\processid}{\ppprocessid}}), \futcontext''_p, \markcontext(\futcontext_{\parallelcomp{\pprocessid}{\ppprocessid}}), \futcontext''_q,\futcontext'_u$, which is trivially equal.
\end{proof}

\begin{definition}\label{def:rdtg}
    \textbf{A ready to go process} is a process $\processid$ with the shape: 
    \begin{itemize}
        \item $\processid = \activethread{\activationrecord}{\expression}{\futid}{\actorid}$
        \item $\processid = \queuedthread{\activationrecord}{\expression}{\futid}{\actorid}$
        \item $\processid = \callmsg{\futid}{\actorid}{\methid}{\many{\val}}$
    \end{itemize}
\end{definition}

\begin{lemma}[Progress or await]\label{lem:await}
    If $\config=\Conf\aactorenvironment\processid$ and $\WTconf{\actorcontext}{\futcontext}{n}{\config}{\ffutcontext}$ and $\config\not\red$, then 
    $\processid \Pcong \pprocessid_1 \parP \activethread\activationrecord\expression{\futid'}{\actorid}  \parP \pprocessid_2$ and $\expeval\activationrecord\expression \expstep[\placefut{\futid}{\val}]{} $ for some $\futid \in \dom(\futcontext)$. 
\end{lemma}
\begin{proof}
    We pattern match $\processid \Pcong \parallelcomp{\pprocessid_1}{\parallelcomp{\processid'}{\pprocessid_2}}$, 
    such as $\parallelcomp\aactorcontext\pprocessid_1$ is terminated (this is always possible by a neutral $\pprocessid_1$), $\processid'$ is a \textit{ready to go process}, 
    and $\pprocessid_2$ is any process (also possibly neutral).
    By linearity of future contexts on the well-typeness hypothesis we know H1:
    $\WTproc{\actorcontext}{\futcontext, \ffutcontext}{m}{\processid'}{\futcontext'}$.
    Split 3 cases for $\processid'$: 
    (1) $\processid'=\callmsg{\futid}{\actorid}{\methid}{\many{\val}}$: by well-formness of $\config$ there must be either an $\idleactor{\actorid}$ in $\config$, in which case we can take congruence steps to move $\idleactor{\actorid}$ and perform (\procspawn) contradicting the hypothesis, or there exists an active thread $\threadid=\activethread{\activationrecord}{\expression}{\futid''}{\actorid}$ in $\pprocessid_2$. If $\threadid$ can reduce it contradicts the hypothesis, else by \cref{lem:stuckthread} $\expeval{\activationrecord}{\expression} \expstep[\placefut{\futid}{\val}]{} \expeval{\activationrecord'}{\expression'}$, which satisfies the premise of (\procsuspend) which contradicts the hypothesis.
    (2) $\processid'=\queuedthread{\activationrecord}{\expression}{\futid'}{\actorid}$: then similarly it is not possible that $\actorid$ is executing another thread, so it must be that $\idleactor{\actorid}$ is in $\config$, and by $\Pcong$ it can activate $\processid'$ turning it into case (3).
    (3) $\processid'=\activethread{\activationrecord}{\expression}{\futid'}{\actorid}$: by \cref{lem:stuckthread}
    and hypothesis it must be that $\expeval{\activationrecord}{\expression} \expstep[\placefut{\futid}{\val}]{} \expeval{\activationrecord'}{\expression'}$, and by H1 $\futid \in \dom(\futcontext \cup \ffutcontext)$. 
    If $\futid \in \dom(\ffutcontext) \implies \placefut{\futid}{\val}$ is in $\pprocessid_1$, hence by precongruence (\procget) can be applied, contradicting the hypothesis. Then $\futid \in \dom(\futcontext)$.
    
\end{proof}

\UTComm{do i need to formalise what does it mean to be IN a configuration?}
By these results we can state deadlock freedom and its proof follows.

\begin{theorem}[Deadlock Freedom] \label{thm:df} 
    If $\WTconf{\actorcontext}{\emptycontext}{n}{\config}{\futcontext}$ then 
    either $\config$ is terminated 
    or $\config \semanticstep \config'$ for some $\config'$. 
\end{theorem}

\begin{proof}
    If $\config$ is not terminated and $\config\not\red$ applying \cref{lem:await} with $\futcontext=\emptycontext$ 
    gives us $\futid \in \dom(\emptycontext)$ which is not possible. Therefore, $\config$ cannot be terminated and have no step at the same time.
\end{proof}


The result of deadlock freedom guarantees that we can always take a step, and to prove weak termination we need to show that not only we can take a step, but we can take a step towards termination.
This is where the typing measure comes into play, and therefore we show in \cref{lem:termconf} the base case that any configuration typed under a 0 measure is terminated.

\begin{lemma}[Terminated configuration]\label{lem:termconf}
  If $\actorcontext \joinctx \futcontext \vdash^{0} \assocctx{\config}{\futcontext'}$, then $\config$ is terminated.
\end{lemma}

\begin{proof}
    Given $\config=\Conf{\aactorenvironment}{\processid}$ we apply induction on the typing rules over processes:
    \textbf{Base cases:} T-result and T-idle are terminated configurations. 
    T-msg and T-thread cannot be applied due to having a measurement strictly bigger to 0.
    \textbf{Inductive case:} (T-parallel), where $\config=\Conf{\aactorenvironment}{\parallelcomp{\processid}{\pprocessid}}$.
    	If any of $\processid$ or $\pprocessid$ is neither $\idleactor{\actorid}$ or $\placefut{\futid}{\val}$, they cannot be a
    	call message nor a thread (active or inactive), because they cannot be typed with a 0 measurement. 
    	Then, $\processid$ and $\pprocessid$ are either both $\idleactor{\actorid}$ or $\placefut{\futid}{\val}$, or at least one 
    	is a parallel composition of processes, in which case we apply inductive hypothesis.
    	Finally, (T-conf) replicates the measurement of the process, therefore a configuration typed with a 0 measurement 
    	holds a process typed with a 0 measurement, which is terminated.
\end{proof}



To isolate the only branching case of the reduction we consider the silent choice of $\nchoicesym$.

\begin{lemma}[Silent choice]\label{lem:silentchoice}
    If $\WTthread\actorcontext\futcontext\typingcontext{n}\actorid{\expeval{\activationrecord}{\e}}\T\actorcontext'$ and 0<n,
    if there exists a step
    $\expeval{\activationrecord}{\e} \expstep[]{\emptylabel} \expeval{\activationrecord'}{\e'}$,
    there exists $\expeval{\activationrecord}{\e} \expstep[]{\emptylabel} \expeval{\activationrecord''}{\e''}$
    such as $\WTthread\aactorcontext\ffutcontext\ttypingcontext{m}\actorid{\expeval{\activationrecord''}{\e''}}\T\aactorcontext'$, and 
    m<n.
\end{lemma}

\begin{proof}
    Trivially holds for (\explet) and (\expop). 
    For (\expchoiceleft) and (\expchoiceright) holds by (\tychoice),
    where the measure is $1+p$ with $p$ the typing measure of 
    one branch.
    We select that branch and return the corresponding step. 
\end{proof}

Helpful Direction uses deadlock freedom and silent step to find the step that decreases the measure.

\begin{lemma}[Helpful Direction]\label{lem:help}
    If $\WTconf[]{\actorcontext}{\emptycontext}{n}{\config}{\futcontext'}$, with 0<n.
    Then there exists $\config\semanticstep\config'$ such as 
    $\WTconf[]{\aactorcontext}{\emptycontext}{m}{\config'}{\ffutcontext'}$ with m<n.
\end{lemma}

\begin{proof}
    By \cref{thm:df} (deadlock freedom) there exists a step. 
    By induction on that step it is trivial to show that all steps reduce the measure, except the silent step.
    By silent choice \cref{lem:silentchoice} there exists $\expeval{\activationrecord}{\e} \expstep[]{\emptylabel} \expeval{\activationrecord'}{\e'}$ such as the measure decreases.
    By (\tythread) the measure of the typing for the configuration decreases.
\end{proof}

Finally, by the theorems proven in this section we can conclude with the proof of weak termination of well-typed configurations.

\begin{theorem}[Weak termination]
  If $
    \WTconf[]{\actorcontext}{\emptycontext}{n}{\config}{\futcontext'}$ then 
    it exists a reduction $\config \semanticstep^{*} \config'$ such as $\config'$ is terminated.
\end{theorem}
\begin{proof}
    We prove by induction on the measurement n. Base case: $n=0$ $\config$ is terminated. Inductive case: $n=m+\mu$ with $\mu>0$ by \cref{lem:help} there exists a step $\config\semanticstep\config'$ such as $\config'$ is typed under a $m<n$ measure.
    By I.H. on $\config'$ there exists a $\config \semanticstep \config' \manysemanticsteps \config''$ such as $\config''$ is terminated.
\end{proof}


\end{document}